\documentclass[3p]{elsarticle} 
 \myfooter[C]{}
\renewcommand{\thispagestyle}[1]{}
\usepackage{booktabs}

\date{}
\def\texpsfig#1#2#3{\vbox{\kern #3\hbox{\includegraphics{#1}\kern #2}}\typeout{(#1)}}
\usepackage{url,hyperref}
\hypersetup{
    colorlinks=true,
    linkcolor=blue,
    filecolor=magenta,      
    urlcolor=cyan,
    pdfpagemode=FullScreen,
    }
\usepackage{bbm}
\usepackage{eurosym}                           
\usepackage[latin1]{inputenc}                  
\usepackage{url}                               
\usepackage{longtable}                         
\usepackage{array}                             
\usepackage{graphicx,color}
\usepackage{amsthm}
\usepackage{amsbsy}
\usepackage{amssymb}
\usepackage{amsmath}
\usepackage{enumerate}
\usepackage{graphicx,color}
\usepackage{epsfig}
\usepackage{multirow,bigdelim}
\usepackage[table]{xcolor}
\usepackage{natbib}
\usepackage{float}
\usepackage{algorithm} %
\usepackage{algpseudocode}
\usepackage{mathtools}
\usepackage{thmtools}
\usepackage{cleveref}
\theoremstyle{plain}
\newtheorem{theorem}{Theorem}[section]
\newtheorem{cor}{Corollary}[section]
\newtheorem{prcd}{Procedure}
\newtheorem{dfn}[theorem]{Definition}

\newtheorem{example}{Example}[section]

\newtheorem*{rem}{Remark}
\theoremstyle{remark}

\theoremstyle{plain}
\newtheorem{lem}[theorem]{Lemma}

  {\par\noindent\textit{Proof of the claim#1.}\quad}
  {\hfill$\square$\par}

\theoremstyle{definition}

\newcommand{\e}{{\rm e}}        
\def\R{\mathbb{ R}}             
\def\vix{X}    
\def\E{\mathbb{ E}}             
\def\Q{\mathbb{ Q}}  
\def\N{\mathbb{ N}}  

\def\pol{\mathcal{P}}

\def\F{\mathcal{F}}  

\newcommand{\mname}{VDV}
\newcommand{\modelname}{VIX-derived volatility }
\DeclarePairedDelimiter{\norm}{\lVert}{\rVert} 
\renewcommand{\d}{{\rm d}}      
\def\dW{{\rm d}W}               
\def\dt{{\rm d}t}

\DeclareMathOperator*{\arginf}{arg\,inf}
\def\ds{{\rm d}s}
\def\du{{\rm d}u}

\def\dS{{\rm d}S}
\def\dv{{\rm d}v}

\newcommand{\abs}[1]{| #1 |}

\newcommand{\bs}[1]{\boldsymbol{#1}}
\def\1{{\mathbbm{1}}}            

\theoremstyle{plain}

\usepackage[margin=1cm]{caption}
\numberwithin{equation}{section}	     

\title{\raggedright The VIX-Derived Volatility Model: A VIX-first Joint SPX-VIX Framework 
}

\begin{document}
\author[1,3]{NICOLA F. ZAUGG \corref{cor1}}
\ead{N.F.Zaugg@uu.nl}
\author[1,2]{\raggedright LECH A. GRZELAK}
\ead{L.A.Grzelak@uu.nl}
\cortext[cor1]{Corresponding author.}
\address[1]{Mathematical Institute, Utrecht University, Utrecht, the Netherlands}
\address[2]{Financial Engineering, Rabobank, Utrecht, the Netherlands}
\address[3]{LGT Bank, Zurich, Switzerland}

\begin{abstract} 
We propose the \emph{\modelname(VDV)} model, a VIX-first framework for joint SPX-VIX modeling. In the model, we define explicit dynamics for the VIX process to price VIX futures and options, yielding a VIX-side calibration that is independent of the SPX dynamics. Using the rolling-window definition of the VIX, we then derive a coupling function to obtain the SPX volatility process as a latent process consistent with the calibrated VIX dynamics. This produces a stochastic-volatility representation for SPX that can be further calibrated to SPX options without altering the VIX dynamics. Relative to existing joint-calibration approaches based on global optimization or highly flexible black-box dynamics, the method offers an interpretable and tractable decomposition of the joint problem by separating the calibration of the VIX from the calibration of the SPX. In a numerical experiment, we show that a \mname~model with local volatility and mean-reverting dynamics for the VIX achieves a close model fit to the VIX futures market, VIX option market, and the SPX option market at the same time, providing a consistent joint framework. 
\noindent 
\end{abstract}

\maketitle
{\let\thefootnote\relax\footnotetext{The views expressed in this paper are the author's personal views and do not necessarily reflect the views or policies of their current or past employers. The authors have no competing interests.}}
\title{VIX as Forward Implied Variance and Options on VIX}
\author{}
\date{}
\section{Introduction}
The joint modeling of S\&P~500 (SPX) and VIX derivatives presents a persistent challenge in quantitative finance. The \emph{CBOE Volatility Index} (VIX) is a financial index measuring the expected future volatility of the SPX, derived from the SPX option market . The definition of the VIX links the spot price of the VIX to index options on the SPX. These SPX options, on the other hand, are sensitive to the VIX volatility , i.e., the volatility of the volatility of the SPX. Since both the VIX and the SPX have actively traded index options, joint modeling requires a model that fits both option markets, as well as the VIX futures market. The structural interdependence complicates the construction of models that consistently reproduce both markets, as each market component requires a separate calibration. This is especially challenging for short maturities where the SPX skew is steep, and VIX smiles are relatively flat \citep{JacquierMartiniMuguruzaRoughBergomiVIX}.

Traditional approaches to the joint problem rely on parametric stochastic volatility models, such as Heston-type or rough volatility frameworks, which attempt to capture both SPX and VIX dynamics within a unified specification \citep{JacquierMartiniMuguruzaRoughBergomiVIX, abi2025joint}. These models often face calibration tensions and require careful tuning to reconcile the two surfaces. More recent nonparametric methods, including martingale optimal transport and Schr\"odinger bridge constructions, offer greater flexibility but typically involve global optimization across marginal distributions \citep{GuoLoeperOblojWang2021, Guyon2024FinanceStoch}. Neural SDEs have also been proposed to jointly calibrate SPX and VIX surfaces with high flexibility \citep{guyon2022neural}.

In this work, we propose a modular framework that reverses the conventional modeling order. We begin by defining the VIX index dynamics as a separate model that can be calibrated to VIX option prices and VIX futures. This yields a one-factor diffusion with transition density and drift-volatility coefficients inferred from market data. Leveraging the rolling-window definition of VIX \citep{CBOEVIXOptions}, we derive a backward equation for the conditional expectation of integrated variance. The backward equation reduces the specification of a coupling map between the VIX index and a volatility process to a single interval, called the terminal interval. We show how to derive this coupling map for \emph{polynomial processes}, and how to approximate it for general, non-polynomial processes. The framework generalizes stochastic volatility models with a single driver for volatility and VIX. It yields a transparent and tractable framework for joint SPX-VIX modeling, grounded in observable market data and realistic VIX and volatility processes. 

\subsection{Definition of a Joint SPX-VIX framework}
\label{sec:def_spxvix}
We consider the price process $(S_t)_{t \geq t_0}$ of the discounted and dividend-adjusted SPX on a risk-neutral probability space $(\Omega, \F, \Q)$ with filtration $\mathbb{F} = \{\F_t, t\geq t_0\}$. Its dynamics are taken to be
\[
\frac{\dS_t}{S_t} = \sigma_t\, \dW_t,
\qquad v_t := \sigma_t^2,
\]
where $W$ is a Brownian motion and $v_t$ is the \emph{instantaneous variance} of the process. By It\^o's formula, we derive the log-transformation of the process as 
\[
\log \frac{S_T}{S_{t_0}}
= -\tfrac12\int_{t_0}^T v_u\,\du + \int_{t_0}^T \sigma_u\,\dW_u.
\]
Taking conditional expectation eliminates the martingale term:
\begin{equation}\label{eq:log_contract}
\mathbb{E}\!\left[\log\frac{S_T}{S_{t_0}} \mid  \F_{t_0}\right] = -\tfrac12\,\mathbb{E}\!\left[\int_{t_0}^T v_u\,\du \mid \F_{t_0}\right].
\end{equation}
The VIX index emerges naturally as the square root of the average expected variance over a 30-day window. Throughout the paper, we denote the constant $\Theta=30/365$ as the 30-day time horizon in years. The forward average variance over an interval $[t,t+\Theta]$ for $t \geq t_0$ is defined as 
\[
\bar v_t := \frac{1}{\Theta}\int_t^{t+\Theta} v_u\,\du.
\]
Applying \eqref{eq:log_contract} shows that the conditional expectation of the average variance can be written as the log-transformed SPX process:
\[
\mathbb{E}[\bar v_t \mid \F_t]
= -\frac{2}{\Theta}\,
\mathbb{E}_{t}\!\left[\log\frac{S_{t+\Theta}}{S_t}\mid \F_t \right].
\]
We define the VIX process\footnote{Note that in practice, the VIX is scaled by a factor of 100. Without loss of generality, we will assume an unscaled  VIX process.}$(\vix_t)_{ t\geq t_0}$ as the $\F_t$-adapted process of the square root of this forward expected variance, annualized and scaled:
\begin{equation}
\label{def:VIX}
\vix_t
:= \sqrt{\mathbb{E}[\bar v_t \mid \F_t]}
= \sqrt{-\tfrac{2}{\Theta}\,
\mathbb{E}\!\left[\log\frac{S_{t+\Theta}}{S_t}\mid \F_t\right]}.
\end{equation}
The VIX index is actively being traded through a futures market. Furthermore, there is a highly liquid market for European options on the VIX index. A European call or put with strike $K \in \R_+$ and maturity $T \geq t_0$ has payoff 
\[(\vix_T-K)^+ \quad \text{or} \quad (K-\vix_T)^+,\]
and must have an arbitrage-free price at time $t \leq T$ under the risk-neutral measure $\mathbb{Q}$. A reasonable joint model of $(S_t, \vix_t)$ therefore must be calibrated to both the SPX options market, as well as the VIX options and futures market, where the VIX dynamics are given by \Cref{def:VIX}.
The pricing formulas reveal the distributional dependencies of $(S_t)_{t\geq t_0}$ when valuing SPX and VIX options. For European options on the SPX maturing at a time $T > t_0$, the option price only depends on the distribution of the variable 
\[
Z_T := \log S_T,
\]
since plain calls and puts depend only on $S_T$.  
For VIX options, by rewriting the settlement index as
\[
\vix_T^2 = -\frac{2}{\Theta}\,\Big(Y_T-Z_T\Big),
\qquad
Y_T=\mathbb{E}\!\big[\log S_{T+\Theta} \mid \F_T\big],
\]
it is clear that valuation of an $X_T$-claim depends on the joint distributional properties of the $\F_T$-random variable $(Z_T,Y_T)$.  

We summarize these insights with the definition of a joint framework:
\begin{dfn}[Joint SPX VIX Framework]
    Let $(S_t)_{t\geq t_0}$ be a model process for the SPX and let $(\vix_t)_{t\geq t_0}$ be given by (\ref{def:VIX}). We call $(S_t, \vix_t)_{t\geq t_0}$ a \emph{joint framework} for the SPX and VIX if the model prices of the following $T$-claims match the observed market prices for all expiries $T \geq t_0$:
    \begin{enumerate}[i)]
    \itemsep0em
        \item SPX European options: \quad $(S_T-K)^+ $ or $(K-S_T)^+$
        \item VIX European options: \quad $(\vix_T-K)^+ $ or $ (K-\vix_T)^+$ 
        \item VIX Futures contracts: \quad $(\vix_T-K)$
    \end{enumerate}
    
\end{dfn}
A consistent joint model of the SPX and VIX derivatives markets is not solely a theoretically motivated challenge, but is of practical importance for hedging and risk management, since market participants routinely use SPX options to hedge VIX exposure and vice versa~\cite{bardgett2019inferring}. For instance, portfolio managers use VIX options for tail-risk hedging of SPX portfolios.
\subsection{Existing Research}
The joint SPX-VIX calibration problem has sparked extensive research in mathematical finance over the last two decades. The most common approaches attempt to apply classical stochastic volatility models to the SPX-VIX problem. These methods specify dynamics for the instantaneous variance process $(v_t)_ {t\geq t_0}$ as a Markovian semimartingale, which then flows into the definitions of $(S_t)_{t \geq 0}$ and $(\vix_t)_{t \geq 0}$. If $(v_t)_ {t\geq t_0}$ is flexible enough, the model can then be calibrated to SPX options, VIX options, and VIX futures. Various articles examined classical stochastic volatility models for $(v_t)_ {t\geq t_0}$, or adaptations thereof. Gatheral considered a CEV model~\cite{gatheral2008consistent} for a joint model, Baldeaux and Badran~\cite{baldeaux2014consistent} utilized a 3/2 model with jumps, and Fouque and Saporito~\cite{fouque2018heston} a generalization of the Heston model. Grzelak~\cite{grzelak2026randomization} considered randomizations of affine diffusion processes, and Ballotta et al. considered time-changed L\'evy processes~\cite{ballotta2025term}. These models were only partially successful at fitting to the market. A possible explanation for this phenomenon was shown by Guyon et al~\cite{guyon2020inversion}, where they explored a general phenomenon for stochastic volatility models and posed a necessary condition for such models to fit the market, called the \emph{inversion of complex order}. 

Since classical stochastic volatility models struggle with this condition and are therefore not general enough to fit both markets, researchers extended the models in various directions. Abi-Jaber et al.~\cite{abi2025joint} introduced \emph{Gaussian Polynomial} volatility models, and Cuchiero et al.~\cite{cuchiero2025joint} utilize signature models to define the volatility process. Alternative approaches to stochastic volatility models for the joint problem exist. In the \emph{martingale optimal transport} approach, the problem is formulated as an optimal transport problem in discrete time~\cite{GuoLoeperOblojWang2021, Guyon2024FinanceStoch}. 
Although the novel approaches produce promising numerical results, a common issue with these models is complexity. The instantaneous variances are modeled with increasing complexity by adding additional stochastic drivers or functionals. This increases both model and numerical complexity, leading to long calibration times, unstable parameters, and results that are difficult to interpret in practice.

Lastly, a similar method to our proposed framework was introduced by Papanicolaou~\cite{papanicolaou2022consistent}. The model derives a stochastic volatility process from a factor model, that is simultaneously used to calibrate to the VIX futures market. Similarly, Papanicolaou shows how to derive a coupling function between the processes. The approach, however, focuses on VIX futures market models and not VIX options, and is restricted to time-homogeneous factors.
\subsection{Proposed Framework}

We propose an intuitive approach to joint SPX-VIX modeling, called the \modelname (\mname) framework. The model is a \emph{stochastic volatility model}: the instantaneous variance of the SPX, $(v_t)_{t\geq t_0}$, is an adapted stochastic process. Unlike classical stochastic volatility models, however, $(v_t)_{t \geq t_0}$ is not modeled explicitly but instead emerges as a latent process derived from the dynamics of the VIX. The VIX process itself is modeled explicitly, allowing seamless calibration to both the VIX futures and options markets. We then derive the corresponding instantaneous volatility dynamics by defining a smooth coupling map $\psi: [t_0,\infty) \times \R_+ \to \R_+$ such that $v_t = \psi_t(\vix_t)$. This map depends on the specification of $(\vix_t)_{t \geq t_0}$ and can be derived explicitly for polynomial processes, a specific process class. In the general, non-polynomial case, we derive a numerical approximation of the map based on the same principles.

Utilizing the mean-reverting local volatility framework following the work of Drimus and Farkas~\cite{drimus2013local}, we introduce the Local Volatility VIX model, an implementation of the VDV model. The explicit dynamics of the VIX process $(\vix_t)_{t\geq t_0}$ allow us to fit VIX options and futures with distinct parameters for each component. We then derive the SPX dynamics that are calibrated to the SPX option market. The numerical results show that the model generates realistic mean-reverting dynamics for the VIX, and is flexible enough to reprice the SPX options market.

The remainder of the paper is organized as follows. Having introduced the VIX and the joint VIX--SPX calibration problem above, we describe in \Cref{sec:model} the general framework of the \mname~model for capturing the joint $\Q$-dynamics of the SPX and the VIX. We then introduce the properties of the coupling function linking the two processes and derive the backward equation for a consistent VDV framework, given by \Cref{eq:backwards_main}. A first example of a VDV model is provided in \Cref{ex:gbm}, illustrating the framework on a simple diffusion process. In \Cref{sec:polynomial} we introduce polynomial processes and show how they are used to derive the coupling function for general processes. An algorithmic description of the full process to derive the coupling function is included in \ref{app:algo}. In \Cref{sec:lv-vix} we introduce the LV-VIX model, a local-volatility-based \mname~model that produces realistic volatility dynamics, and in \Cref{sec:num} we apply this implementation to market data from April 27th, 2026. We conclude the paper in \Cref{sec:conclusion}, followed by \ref{app:lvol}, which contains accompanying theoretical results.

\section{A VIX-first Stochastic Volatility Model}
\label{sec:model}
\subsection{General Model Dynamics}
The \emph{\mname} model is a joint SPX/VIX framework emerging from an explicit stochastic model for the VIX process. Let $(\Omega, \F, \Q)$ be a probability space with filtration $\F_t, t\geq t_0$. We define the discounted SPX process $(S_t)_{t\geq t_0}$ and the VIX process $(\vix_t)_{t\geq t_0}$ through a system of two differential equations linked through a function $\psi_t(x)$. The dynamics of $(\vix_t)_{t\geq t_0}$ follow a general It\^o process with drift term $\mu(t,x)$ and diffusion term $\sigma(t,x)$. We consider the following dynamics on $[t_0, \infty)$:
\begin{align}
    \frac{\dS_t}{S_t} &= \sqrt{v_t}\, \dW^{(S)}_t,\label{eq:general_S}\\
    v_t &= \psi_t(\vix_t),\label{eq:general_v}\\
    \d\vix_t &= \mu(t,\vix_t)\,\dt +  \sigma(t,\vix_t)\dW_t^{(\vix)}.\label{eq:general_vix}
\end{align}
The expressions $W^{(\vix)}_t, W^{(S)}_t$ are Brownian motions with a correlation coefficient given by $\dW^{(\vix)}_t\dW^{(S)}_t = \rho \dt$. Since the VIX is a strictly positive quantity, we require $\mu(t,x), \sigma(t,x)$ to be chosen such that $(\vix_{t})_{t \geq t_0}$ is almost surely a positive process. Furthermore, for \Cref{lem:rolling} we require that the functions satisfy the usual linear growth conditions. There exists a constant $K > \R$, such that
\begin{equation}
    \label{eq:linear_growth}
    \abs{\mu(t,x)} + \abs{\sigma(t,x)} \leq K (1+ x), \quad x \in \R_+.
\end{equation}
The function $\psi: [t_0,\infty) \times \R_+ \to \R_+$ is called the \emph{coupling function}, as it couples the VIX at time $t$ to the \emph{instantaneous variance} $(v_t)_{t \geq t_0}$. We study this coupling function in detail in \Cref{sec:coupling}. For the system to be well-defined, we require it to be positive and measurable, such that $(v_t)_{t \geq t_0}$ is an adapted, almost surely positive stochastic process. Although the instantaneous variance process $(v_t)_{t\geq t_0}$ is not explicitly modeled by a stochastic differential equation, one can derive the SDE for $(v_t)_{t\geq t_0}$ through It\^o's lemma under additional assumptions of differentiability of $\psi$.
\subsection{From VIX to Instantaneous Variance - Coupling Function}
\label{sec:coupling}
The essence of the joint model defined in \crefrange{eq:general_S}{eq:general_vix} is the coupling function $\psi: [t_0,\infty) \times \R_+ \to \R_+$, that links the two stochastic processes through the process of the instantaneous variance $(v_t)_{t\geq t_0}$. Although not strictly required, we will assume $\psi_t(x)$ to be smooth in the space components (i.e., $ x \mapsto \psi_t(x) \in C^\infty, \forall t\geq t_0 $), and at least continuous in the time direction. The continuity of $\psi_t$ ensures that $(v_t)_{t \geq t_0}$ is adapted and almost surely continuous. A direct consequence of such a coupling model is that $v_t$ is a $\vix_t$-measurable random variable for all $t\geq t_0$,  since $\psi$ is a deterministic function. Therefore, we have the interpretation of $\psi_t$ as the conditional expectation of $v_t$ given the current level $\vix_t$.
\begin{equation}
    \E\left[v_t \,\middle| \,\vix_t\right] = v_t = \psi_t(\vix_t).
\end{equation}
Here, $\E\left[v_t \,\middle| \,\vix_t\right]$ represents the conditional expectation of $v_t$ given the $\sigma$-algebra $\sigma(\vix_t)$ generated by $\vix_t$.

Under the positivity assumptions, the general model of \crefrange{eq:general_S}{eq:general_vix} is well-defined. However, since the VIX-process is also derived from the SPX as described in \Cref{sec:def_spxvix}, the choice of $\psi_t$ is not free, as the definition of the VIX through \Cref{def:VIX} must be fulfilled for all times $t\geq t_0$. The additional condition on the interaction of $(\vix_t)_{t \geq t_0}$ and $(S_t)_{t \geq t_0}$ further restricts the choice of $\psi_t$. We introduce the concept of \emph{consistency} of the coupling function $\psi$.
\begin{dfn}[Consistency]
\label{def:consistent}
We call the model \crefrange{eq:general_S}{eq:general_vix} consistent if for all $t\geq t_0$, we have
\begin{equation}\label{eq:consistency-equation}
\vix_t^2 = \frac{1}{\Theta}\,
\mathbb{E}\!\left[\int_t^{t+\Theta} v_u\,\du \;\middle|\; \mathcal{F}_t\right], \quad \text{a.s.}
\end{equation}
In this case, we call $\psi$ a consistent coupling function. \Cref{eq:consistency-equation} is called the \emph{consistency equation}
\end{dfn}
The definition of consistency poses the questions of existence and uniqueness of a consistent $\psi$ given a process $(\vix_t)_{t \geq t_0}$. To answer these questions, we first assume that a consistent $\psi$ exists. We will now show that, under the assumption of consistency, the specification of $\psi$ can be restricted to a shorter time interval $[T,T+\Theta]$ of length $\Theta$, as $\psi_t$ outside of the interval can be uniquely traced back to the definition of $\psi_t$ inside the interval. This temporal dependency evolves as a necessary condition of \Cref{eq:consistency-equation} through the Feynman-Kac theorem. We first apply the Feynman-Kac theorem to the right-hand side of \Cref{eq:consistency-equation} to obtain a partial differential equation for the conditional expectation.
\begin{lem}[Rolling-window backward equation]
\label{lem:rolling}
Suppose the dynamics of a stochastic process $(\vix_t)_{t\geq t_0}$ solve the differential equation \[\d\vix_s = \mu(s, \vix_s) \, \ds + \sigma(s, \vix_s) \, \dW_s^{(\vix)},\] where $\mu(t,x), \sigma(t,x)$ are continuous and satisfy the linear growth condition, and let
\[
I_t(x) := \mathbb{E} \left[ \int_t^{t+\Theta} \psi_u(\vix_u) \, \du \,\middle|\, \vix_t = x \right], \qquad
\phi_t(x) := \mathbb{E} \left[ \psi_{t+\Theta}(\vix_{t+\Theta}) \,\middle|\, \vix_t = x \right].
\]
Suppose that for all $t$ the function $\psi_t(x)$ is smooth and $\psi_t(\vix_t)$ is integrable. Then, for any $t \geq t_0$, the expression $I_t(x)$ satisfies:
\begin{equation}
\label{eq:diff_1}
\partial_t I_t(x) + \mu(t, x) \partial_x I_t(x) + \tfrac{1}{2} \sigma^2(t, x) \partial_{xx} I_t(x) + \psi_t(x) - \phi_t(x) = 0.
\end{equation}
\end{lem}

\begin{proof}
Fix $s \geq t_0$ and consider the auxiliary function:
\[
W(s,t, x) := \mathbb{E} \left[ \int_t^{s} \psi_u(\vix_u) \, \du \,\middle|\, \vix_t = x \right].
\]
Note that $W(t,t, x) = 0$ and our target function is $I_t(x) = W(t+\Theta,t, x)$. Since $\phi$ is bounded and measurable, the function $(t,x) \mapsto W(s,t, x)$ is of the class $C^{1,2}$. By the Feynman-Kac backward relation for a running cost with a terminal condition of zero at $t=s$, we obtain
\[
(\partial_t + \mathcal{L}_\vix) W(s,t, x) + \psi_t(x) = 0, \qquad W(s,s, x) = 0,
\]
where $\mathcal{L}_\vix f(t, x) := \mu(t, x) \partial_x f(t, x) + \tfrac{1}{2} \sigma^2(t, x) \partial_{xx} f(t, x)$. 

By the integrability of $\psi$, we can differentiate the function $W(s,t,x)$ with respect to its first argument, and we obtain, using the Leibniz rule,
\[
\partial_s W(s,t, x) |_{s=s} = \mathbb{E} \left[ \psi_s(\vix_s) \,\middle|\, \vix_t = x \right].
\]
Now, we apply the operator $(\partial_t + \mathcal{L}_\vix)$ to the function $I_t(x) = W(t+\Theta,t, x)$, using the chain rule for the moving horizon $s(t) = t + \Theta$:
\[
(\partial_t + \mathcal{L}_\vix) I_t(x) = \left[ (\partial_t + \mathcal{L}_\vix) W(s,t, x)  \right]_{s=t+\Theta} + \left[ \partial_s W(s,t, x) \frac{\ds}{\dt} \right]_{s=t+\Theta}.
\]
Since $\frac{\ds}{\dt} = 1$, substituting our previous identities yields:
\[
(\partial_t + \mathcal{L}_\vix) I_t(x) = -\psi_t(x) + \mathbb{E} \left[ \psi_{t+\Theta}(\vix_{t+\Theta}) \,\middle|\, \vix_t = x \right].
\]
Recalling the definition $\phi_t(x) := \mathbb{E} [ \psi_{t+\Theta}(\vix_{t+\Theta}) | \vix_t = x ]$, we arrive at:
\[
\partial_t I_t(x) + \mu(t, x) \partial_x I_t(x) + \tfrac{1}{2} \sigma^2(t, x) \partial_{xx} I_t(x) + \psi_t(x) - \phi_t(x) = 0.
\]
This proves the claim.
\end{proof}
The differential equation \Cref{eq:diff_1} describes the time dynamics of $I_t(x)$. Since we defined $v_t$ in our model to be $\sigma(\vix_t)$-measurable, we can derive that under a consistent model, $I_t(x)$ describes the (scaled) squared VIX process. From the definition of a consistent model, we apply the tower property and find 
\begin{align*}
    \Theta \vix_t^2 &= \Theta \E \left[ \vix_t^2\,\middle|\, \vix_t \right] \\
    &= \E \left[ \mathbb{E}\!\left[\int_t^{t+\Theta} v_u\,\du \;\middle|\; \mathcal{F}_t\right]\,\middle|\, \vix_t \right]\\
    &= \E\!\left[\int_t^{t+\Theta} v_u\,\du \;\middle|\; \vix_t \right]\\
    &= I_t(\vix_t).
\end{align*}
In a consistent model, $I_t(\vix_t)$ is explicitly given by $\Theta \vix_t^2$, and we can apply the newly derived backward equation to this formula. This yields an explicit temporal dependency of $\psi_t$.
\begin{cor}[Consistency relation]
In a consistent model we have $I_t(x) = \Theta x^2$ and therefore we find
\begin{equation}
    \label{eq:backwards_main}
    \psi_t(x)=\phi_t(x)-\Theta\big(2x\,\mu(t,x)+\sigma^2(t,x)\big).
\end{equation}
\end{cor}
\begin{proof}
With $I_t(x) = \Theta x^2$ we have $\partial_t I=0$, $\partial_x I=2\Theta x$, and $\partial_{xx}I=2\Theta$. Plugging into \Cref{lem:rolling} yields
\[
0=\Theta\big(2x\,\mu(t,x)+\sigma^2(t,x)\big)+\psi_t(x)-\phi_t(x),
\]
which rearranges to the stated identity.
\end{proof}

The backward equation is not a definition of $v_t = \psi_t(\vix_t)$. Rather, it maps the definition of $v_{t+\Theta}$ to the $v_t$ for the fixed time $\Theta = \frac{30}{365}$. In particular, given the process $\vix_t$, if we define $v_t$ on a domain $[T, T+\Theta]$, we can derive $v_t$ for any $t \leq T$. To see this, let the function $\psi_t(x)$ be given on $[T, T+\Theta]$. We can derive $\psi_t(x)$ for any $t \in [T-\Theta, T)$ through \Cref{eq:backwards_main}, by 
\[  \psi_t(X_t)=\E[\psi_{t + \Theta} (X_{t+\Theta})\mid X_{t}]-\Theta\big(2X_t\,\mu(t,X_t)+\sigma^2(t,X_t)\big).\]
Since we have now derived $v_t = \psi_t(X_t)$ on $[T-\Theta,T+\Theta]$, we can repeat the argument and derive $\psi_t$ for all $t \geq t_0$. This shows that a consistent model is fully determined by its restriction to a single interval of length $\Theta$. For this reason, we introduce the concept of a \emph{terminal interval} $[T_{F}, T_F+\Theta]$, where $T_F > t_0$, over which we will choose the definition of $\psi$ as part of the modeling. We denote the function $\psi_t$ on the terminal interval as a \emph{terminal curve} \begin{equation}
    \bs{\psi}_{T_F} := \{\psi_t \in C^\infty:  t \in [T_F, T_F+\Theta]\}.
\end{equation} 

Although processes $v_t$ and $\vix_t$ exist beyond the terminal interval, we will assume that beyond $T_F+\Theta$, the exact value of $\psi_t$ and $v_t$ are left unspecified. By introducing the terminal curve, we can reduce the construction of $\psi_t$ to a single interval $[T_F, T_F+\Theta]$, after which the process on $[t_0,T_F+\Theta)$ is given by $\vix_t$ and the backward equation in \Cref{eq:backwards_main}. To ensure consistency in the model and well-definedness, we constrain a terminal curve with additional conditions of continuity and consistency:
\begin{dfn}[Consistent Terminal Curve]
\label{def:consistency}
A \emph{consistent terminal curve} for a terminal time $T_F \geq t_0$ is the family of functions $\bs{\psi}_{T_F} = \{\psi_t \in C^\infty:  t \in [T_F, T_F+\Theta]\}$, such that
    \begin{enumerate}[i)]
\itemsep0em
    \item $\psi_{T_{F}}(x), \psi_{T_F+\Theta}(x)$ fulfill the backward relation
    \item The map $t \mapsto \psi_t$ is continuous on $[T_F, T_F+\Theta]$
\item The consistency equation is fulfilled at $t = T_{F}$: \begin{equation*}
    \Theta X_{T_{F}}^2 = \mathbb{E}\!\left[\int_{T_{F}}^{{T_{F}}+\Theta} \psi_{u}(\vix_u)\,\du \;\middle|\; X_{T_{F}}\right], \quad \text{a.s.}
\end{equation*}
\end{enumerate}
\end{dfn}
These conditions ensure that the previous definition of consistency is fulfilled across the entire interval $[t_0,T_F+\Theta]$. While Condition ii) is a basic assumption of $\psi$ across time, Condition i) ensures that the previously established backward relation is also true at the terminal time. As we will show below, Condition iii) suffices to ensure that the model remains consistent across the entire time interval. 

The three requirements for a consistent terminal curve are relatively weak and allow for a large range of terminal curves. Defining a function $\psi_{T_F+\Theta}(.)$ specifies $\psi_{T_{F}}(.)$ through Condition i). Since these two functions are homotopic, any continuous deformation of these functions will fulfill Conditions i) and ii). It then remains to deform them in a way such that property iii) is fulfilled, which is non-unique.

It remains to show that a consistent terminal curve \emph{remains} consistent as we move backward in time
\begin{lem}[Backward Consistency]
\label {lem:consistent}
    Let $(\vix_t)_{t \geq t_0}$ be a given VIX process in a VDV framework, and consider a terminal curve $\bs{\psi}_{T_F} = \{\psi_t \in C^\infty : t \in [T_F, T_F+\Theta]\}$ which is consistent.
    Then, for all $t \leq T_F$, we have
    \[\vix_{t}^2 = \frac{1}{\Theta}\E\left[\int_{t}^{t+\Theta} \psi_u(\vix_u) \du \,\bigg|\, \vix_{t}\right] .\]
\end{lem}
\begin{proof}
Consider the PDE (\ref{eq:diff_1}) on the interval $[t_0,T_F]$. We claim that if $\psi_t$ follows the backward equation, then $I_t(x) = \Theta x^2$ is a solution to this PDE. 
Since $\partial_t I = 0,\partial_x I = 2x\Theta $ and $\partial_xx I= 2\Theta$, we have
\begin{equation*}
\Theta \left[2 x \mu(t,x) +\sigma^2(t,x)\right] + \psi_t(x) - \phi_t(x) = 0,
\end{equation*}
This means that $I_t(x) = \Theta x^2$ is a solution to this PDE. Given the regularity conditions of the PDE, and since $\Theta x^2$ fulfills the boundary condition by the terminal curve assumption, we conclude by the Feynman-Kac uniqueness theorem under polynomial growth that $I_t(x) = \Theta x^2$ is the unique solution fulfilling the PDE. We can therefore conclude that
    \[\Theta x^2 = \E\left[\int_{t}^{t+\Theta} \psi_u(\vix_u) \du \,\bigg|\, \vix_{t} = x\right], \quad t \in [t_0,T_F],\]
from which we deduce the equation in the claim.
\end{proof}
While the lemma proves that a consistent terminal curve leads to a \emph{consistent} model (as defined by \Cref{def:consistent}), it is unclear if the model is well-defined. The reason is that $\psi_t$ is not guaranteed to be a positive function on $[t_0,T_F)$, even if the terminal curve is chosen entirely from positive functions. To our knowledge, there is no general condition to ensure positivity. For specific classes of $(\vix_t)_{t \geq t_0}$ and $\psi$, it is possible to show the existence of a consistent, positive $\psi$ across the entire interval (see \Cref{ex:gbm} and \Cref{sec:polynomial}).
\subsection{VIX as Geometric Brownian Motion}
\label{ex:gbm}
We now illustrate the theoretical concepts with an example where the VIX process $(\vix_t)_{t\geq t_0}$ follows a driftless geometric Brownian motion (GBM). We consider a \mname~model on $t \in [t_0,T_F+\Theta]$ with the following dynamics:
\begin{align*}
    \frac{\dS_t}{S_t} &= \sqrt{v_t}\dW_t^{(S)},\\
    v_t &= \psi_t(\vix_t),\\
    \d \vix_t &= \sigma \vix_t \dW_t^{(X)},
\end{align*}
where the parameter $ \sigma > 0$ defines the VIX volatility and $\vix_{t_0},S_{t_0}>0$ are the given start values. We aim to derive a consistent coupling function $\psi$ to complete the joint model. The first step is therefore to define a consistent terminal curve $\bs{\psi}_{T_F}$ for a terminal time $T_F$. Since the VIX process can be interpreted as the \emph{expectation of the average future variance $v_t$}, we use a quadratic ansatz $\psi_t(x) = x^2$ on $t \in [T_F,T_F+\Theta]$. We derive the right-hand side of the consistency equation (\ref{eq:consistency-equation}):
\begin{equation}
\label{eq:almost_consistent}
    \frac{1}{\Theta} \mathbb{E}\!\left[\int_t^{t+\Theta} \psi_u(X_u)\,\du \;\middle|\; \mathcal{F}_t\right] = \frac{1}{\Theta} \int_0^\Theta \e^{\sigma^2 u} \du \, \vix_t^2 = \frac{\left(\e^{\sigma^2 \Theta} - 1  \right)}{\sigma^2 \Theta} \, \vix_t^2 .
\end{equation}
This follows from the property of the geometric Brownian motion, where the conditional second moment is given analytically by
\begin{equation}
\label{eq:gbm_pattern}
    \E[\vix^2_{t+\Delta} |\vix_{t} ] = \vix^2_{t}\e^{\sigma^2\Delta},
\end{equation}
for any $\Delta \geq 0$. From \Cref{eq:almost_consistent} we see that the terminal curve is not consistent, as $\frac{\left( \e^{\sigma^2 \Theta} -1  \right)}{\sigma^2 \Theta} \neq 1$. To adjust for the scaling term, we therefore apply a scalar $c >0$ to the definition of $\psi$ and consider the case where the terminal curve is given by \[\psi_{t}(x) := cx^2, \quad c= \frac{\sigma^2 \Theta}{\left(\e^{\sigma^2 \Theta} -1 \right)},\]
for $t \in [T_F, T_F+\Theta]$. With this choice of $c$ we find from \Cref{eq:almost_consistent} that
\[    \frac{1}{\Theta} \mathbb{E}\!\left[\int_{T_F}^{T_F+\Theta} \psi_u(X_u)\,\du \;\middle|\; \mathcal{F}_{T_F}\right] = \frac{c}{\Theta} \int_0^\Theta \e^{\sigma^2 u} \du \, \vix_{T_F}^2 = \frac{c\,\left(\e^{\sigma^2 \Theta} -1 \right)}{\sigma^2 \Theta} \, \vix_{T_F}^2 = \vix_{T_F}^2,\]
leading to consistency at $T_F$. This means that the terminal curve $\bs{\psi}_{T_F}$ fulfills Conditions ii), iii) of a consistent terminal curve. We need to confirm that the endpoints $\psi_{T_F},\psi_{T_F+\Theta}$ fulfill the VIX backward equation. Applying the VIX backward equation in \Cref{eq:backwards_main} to $\psi_{T_F+\Theta}$, using the observation from \Cref{eq:gbm_pattern}, we find
\[\psi_{T_F}(x) =\left[c\e^{\sigma^2\Theta} -\sigma^2 \Theta \right]x^2  = \left[\frac{\sigma^2 \Theta\e^{\sigma^2\Theta}}{\left(\e^{\sigma^2 \Theta} -1 \right)} -\sigma^2 \Theta \right] x^2 = \frac{\sigma^2 \Theta}{\left(\e^{\sigma^2 \Theta} -1 \right)} x^2. \]
This shows that our terminal curve also fulfills Condition i), and we have found a consistent terminal curve.

We can then use the backward equation to compute $\psi$ on $[t_0, T_F)$. Since the VIX process is time-homogeneous, and since $\psi$ is constant in time over the terminal curve, we observe that the backward equation keeps $\psi$ constant until $t_0$, since the definition of $c$ is time-independent. We conclude then that the function $\psi$ such that $\psi_t(x) = c\,x^2$ for all $t \in [t_0,T_{F}+\Theta]$ is a consistent coupling function for the GBM-driven VDV model.

With the complete model, we simulate paths of $(\vix_t, S_t)_{t\geq t_0}$ using a Monte-Carlo procedure, and price VIX options and SPX options from the simulations. We choose $T_F = \frac{11}{12}$ and $\sigma=0.7$, resulting in a value of $c \approx 0.98$. \Cref{fig:gbm-experiment} shows the resulting implied volatility shapes at various time steps. For the VIX options we plot the implied volatilities directly obtained from $(\vix_t)_{t \geq t_0}$ (orange), as well as samples obtained through $(v_t)_{t \geq t_0}$ using a Least-Square Monte-Carlo approach to compute the conditional expectation of the right-hand side of the consistency equation (\ref{eq:consistency-equation}) (blue). Since the VIX process is a GBM, the resulting implied volatility slice is flat at a level of $\sigma$. For the SPX options, we obtain flexibility around the skew of the SPX by varying the correlation parameter $\rho \in [-1,1]$, without altering the distribution of $\vix_t$.
\begin{figure}[H]
    \centering
    \includegraphics[width=1\linewidth]{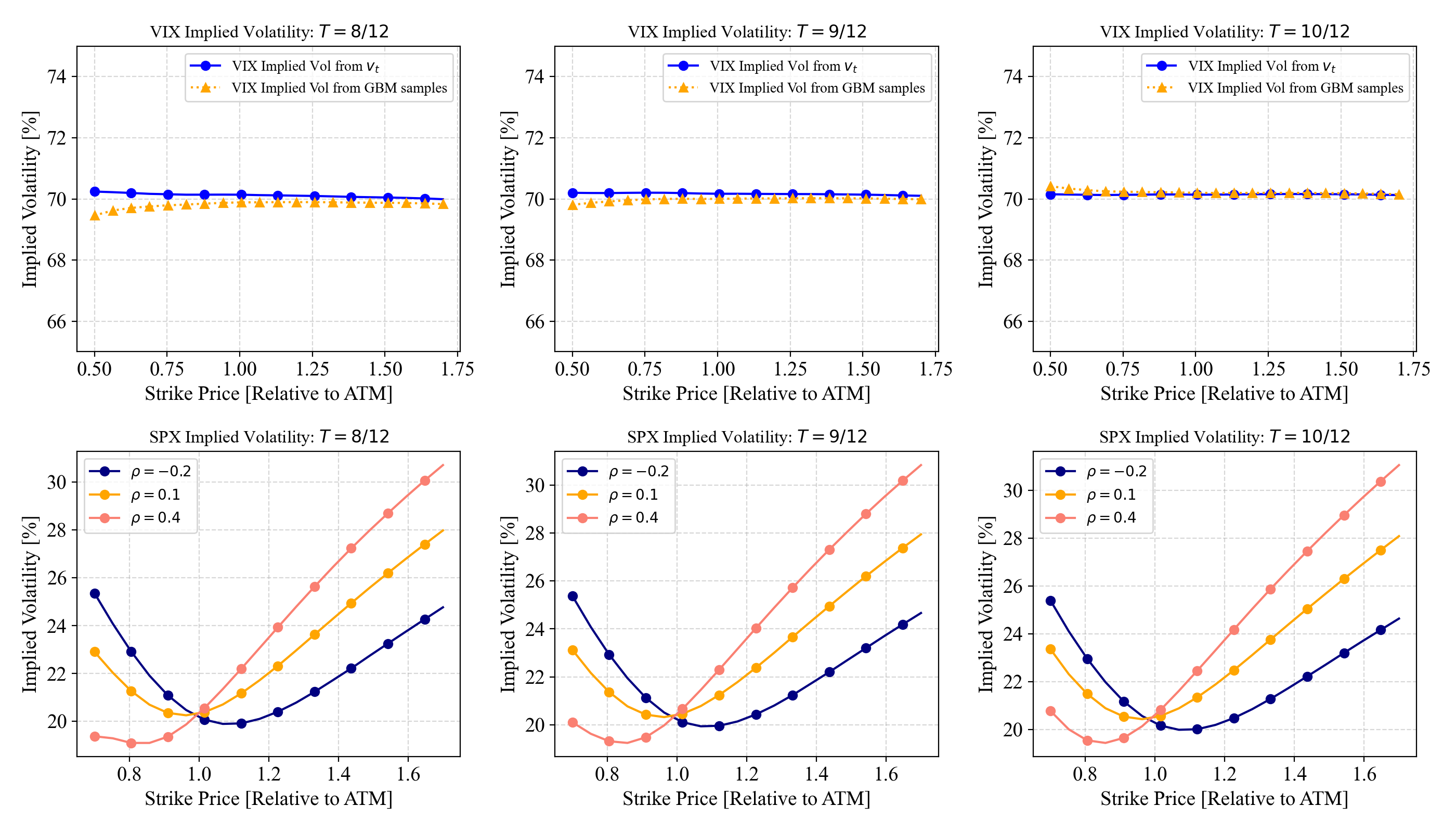}
    \caption{Pricing VIX options and SPX options in the GBM VDV Model: The backward equation leads to consistent paths.}
    \label{fig:gbm-experiment}
\end{figure}

\subsection{Consistency under Polynomial Processes}
\label{sec:polynomial}
The computations of \Cref{ex:gbm} showed that a reasonable ansatz for $\psi_t$ is to consider a function with a quadratic component in $x$, due to the interpretation of $\vix_t^2$ as the \emph{average expected volatility} over a 30-day horizon. In this section, we generalize the result of the GBM to a more general class of processes for $(\vix_t)_{t \geq t_0}$. We introduce \emph{polynomial processes} and show how to derive a consistent terminal curve to complete the VDV model. The theoretical insights of this section are key to approximating a consistent terminal curve for general VIX dynamics as in \Cref{eq:general_vix}.

Polynomial processes are a class of stochastic processes introduced by Cuchiero et al.~\cite{cuchiero2012polynomial} as a generalization of affine processes. The main property to characterize polynomial processes is that the conditional expectation preserves polynomial maps. Initially introduced for time-homogeneous processes, the notion was later generalized to time-inhomogeneous processes~\cite{agoitia2020time}. For our purpose, we define polynomial processes on a one-dimensional state space $\R$. Let $m \in \N$ be a positive integer. Consider the set of functions
\begin{equation}\pol_m  := \left\{ p : \mathbb{R}\to\mathbb{R} \; \middle| \; p(x) = \sum_{k=0}^m \alpha_k x^k,\ \ \alpha_0,\dots,\alpha_m \in \mathbb{R} \right\}, \quad m \in \N.\end{equation}
These sets are well-known to be finite-dimensional subspaces of $C^\infty(\R)$ for any $m \geq 0$. Polynomial processes are simply processes that are invariant on $\pol_m$ under conditional expectation.
\begin{dfn}
   We define an It\^o process $(\vix_t)_{t \geq t_0}$ to be a \emph{polynomial process} of degree $m \in \N$ if for any $p \in \pol_m$ and any $ t_0 \leq t \leq s$, we have
   \begin{equation}
       p_{s,t}(x) := \E\left[p(X_s) \mid X_t = x \right] \implies  p_{s,t}(x)\in \pol_m, \quad
   \end{equation}
\end{dfn}
The reason why polynomial processes are particularly useful for VIX modeling is that the left-hand side of the consistency equation (\ref{eq:consistency-equation}) is also a polynomial. Using linearity of expectation, we can then express the right-hand side as a polynomial as well and determine the coefficients of a consistent terminal curve. To prove this, we require a technical lemma about polynomial processes:
\begin{lem}
\label{lem:tech}
Let $(\vix_t)_{t \ge 0}$ be a polynomial process. For any $t_0 \le t < s$, the conditional expectation $\mathbb{E}[X_s^2 \mid X_t = x]$ is a polynomial in $x$ of degree exactly $2$ and its leading coefficient is strictly positive.
\end{lem}
\begin{proof}
    See \ref{app:proof_lem}
\end{proof}

\begin{theorem}
\label{thm:polynomial}
    Consider a VDV model with the VIX process $(\vix_t)_{t \geq t_0}$ being a polynomial process of degree $ m \geq 2, m \in \N$, and let $[T_F, T_F+\Theta]$ be a terminal time. Then, there exists a polynomial $p \in \pol_m$ of degree $m$, such that
    \begin{equation}
        \vix_{T_F}^2 = \frac{1}{\Theta}\,
\mathbb{E}\!\left[\int_{T_F}^{T_F+\Theta} p(\vix_u)\,\du \;\middle|\; \mathcal{F}_{T_F}\right], \quad \text{a.s.}
    \end{equation}
 
    \begin{proof}
        Consider a polynomial $p(x) := a_0 + a_1x + a_2x^2 + \dots + a_mx^m$ of degree $m$. Using the linearity of the conditional expectation, we find
        \begin{align*}
        \E\!\left[\int_{T_F}^{T_F+\Theta} p(\vix_u)\,\du \;\middle|\; \mathcal{F}_{T_F}\right]
        &= \E\!\left[\int_{T_F}^{T_F+\Theta} \sum_{k=0}^m a_k \vix_u^k \,\du \;\middle|\; \mathcal{F}_{T_F}\right] \\
        &= \sum_{k=0}^m a_k \int_{T_F}^{T_F+\Theta} \E\!\left[\vix_u^k \;\middle|\; \mathcal{F}_{T_F}\right] \du \\
        &= \sum_{k=0}^m a_k \int_{T_F}^{T_F+\Theta} \big(p_{T_F,u, k}\big)(\vix_{T_F}) \,\du,
    \end{align*}
    where $p_{T_F,u, k}$ is a polynomial of $k$ degrees, and the interchange of integration and conditional expectation is justified by Fubini's theorem. Let us denote $b(u;l,k)$ the $l$-the coefficient of the polynomial $p_{T_F,u, k}$, where $l \leq k$. Integrating the coefficients\footnote{The map $u \mapsto b(u; l,k)$ is continuous in $u$ due to the continuity of $\vix_t$} $b(u;l,k)$ from $u = T_F$ to $u = T_F+\Theta$, we denote 
    \[\bar{b}(l,k) = \int_{T_F}^{T_F+\Theta}b(u;l,k) \du.\]
    Then, we find
    \[ \sum_{k=0}^m a_k \int_{T_F}^{T_F+\Theta} \big(p_{T_F,u, k}\big)(\vix_{T_F}) \,\du =\sum_{k=0}^m a_k\sum_{l=0}^k \bar{b}(l,k) \vix_{T_F}^l. \]
    For this expression to be equal to $\Theta\vix_{T_F}^2$, we obtain a system of $m+1$ equations, where
    \[\sum_{k=0}^m a_k\bar{b}(l,k) = \begin{cases}
        \Theta, \quad \text{ if } l = 2,\\
        0, \quad \text{ otherwise. }
    \end{cases}\]
    This is a linear system $Ax = b$ of $m+1$ equations and $m+1$ unknowns. If the matrix $A$ of the linear system is invertible, there is a unique solution. If the matrix is not full-ranked, we have to ensure that a solution exists. \Cref{lem:tech} ensures that the second-order term $\bar{b}(2,2)$ is non-vanishing. Furthermore $\bar{b}(1,1)$ is trivially non-vanishing since $(X_t)_{t \geq t_0}$ is a positive process, meaning that $b = (0,1,0\dots,0)$ is in the column-space of $A$. This concludes the proof.
\end{proof}
\end{theorem}
\begin{rem}
    Since the solution of $Ax =b$ is unique when $A$ is full-ranked, and since there is always a solution $p \in \pol_2$, it implies that
    \[a_k \neq 0, k > 2 \iff \text{rank}(A) < m+1 .\]
    In other words, coefficients of higher order than $2$ are only non-zero if $A$ is degenerate.
\end{rem}
For polynomial processes we can therefore find terminal curves explicitly from the coefficients of the polynomials, as shown in \Cref{thm:polynomial}. The corresponding terminal curve is constant in time, and fulfills Conditions ii) and iii) of a consistent terminal curve. Since Condition i) is not guaranteed, the terminal curve $\psi_t = p$ where $t \in [T_F, T_F+\Theta]$ is not yet a consistent terminal curve. If $(\vix_t)_{t\geq t_0}$ is time-homogeneous, then the terminal curve remains constant under the backward equation, and therefore condition i) is satisfied. If the process is time-inhomogeneous, we can reason to apply the following $\epsilon$-approximation:
\begin{lem}
\label{lem:epsilon}
    Let $\epsilon > 0$ be fixed. Consider the conditions of \Cref{thm:polynomial} and let $p \in \pol_m$ be the corresponding polynomial. Then, there exists a terminal curve $\bs{\psi}_{T_F}$ fulfilling Conditions i) and ii), such that
    \begin{equation}
\norm{\Theta \vix_{T_F}^2 - \mathbb{E}\!\left[\int_{T_F}^{T_F+\Theta} \psi_u(\vix_u)\,\du \;\middle|\; \mathcal{F}_{T_F}\right]}_2 < \epsilon.
    \end{equation}
    \begin{proof}
        See \ref{app:proof_epsilon}
    \end{proof}
    \end{lem}
As a concrete illustration, we consider a second-order polynomial process that is time-homogeneous:
\begin{example}[Brennan--Schwartz Dynamics]
\label{ex:brennan-schwartz}
We consider the following VIX dynamics in a VDV model:
\begin{equation}
\label{eq:brennan-schwartz}
    \d \vix_t = k(\theta -  \vix_t)\,\dt + \sigma \vix_t \, \d W_t,
\end{equation}
where $k,\theta,\sigma$ are fixed constants, such that $k \neq \sigma^2, 2 k\neq \sigma^2 $. This diffusion is a mean-reverting lognormal process often referred to as a \emph{Brennan-Schwartz-type diffusion}~\cite{brennan1980conditional} in the interest-rate literature. The model is time-homogeneous, meaning that the coefficients of the polynomial only depend on the difference $\tau := u - t$. We derive the moments using It\^o's lemma:
\begin{align*}
    \E\left[\vix_u \mid \vix_{t} = x\right]
    &= \e^{-k\tau}\, x \;+\; \theta\big(1-\e^{-k\tau}\big)
    := b(u;\, 1,1)\, x + b(u;\, 0,1),\\
    \E\left[\vix_u^2 \mid \vix_{t} = x\right]
    &= \e^{-(2k-\sigma^2)\tau}\, x^2
    \;+\; \frac{2k\theta}{k-\sigma^2}\Big(\e^{-k\tau} - \e^{-(2k-\sigma^2)\tau}\Big) x \\
    &\quad + \theta^2\left[\frac{2k}{2k-\sigma^2} - \frac{2k}{k-\sigma^2}\e^{-k\tau} + \frac{2k^2}{(k-\sigma^2)(2k-\sigma^2)}\e^{-(2k-\sigma^2)\tau}\right] \\
    &:= b(u;\, 2,2)\, x^2 + b(u;\, 1,2)\, x + b(u;\, 0,2) \nonumber
\end{align*}
We integrate these coefficients over $\tau \in [0,\Theta]$ and write
\begin{equation*}
    \bar b(l,k) := \int_0^\Theta b(\tau; l,k)\,\d\tau, \qquad l \le k \le 2.
\end{equation*}
We now find the coefficients $a_0,a_1,a_2\in\R$ such that
\begin{equation*}
    \psi_t(x) :=  a_0 + a_1x + a_2x^2
\end{equation*}
is consistent at $T_F$ in the sense of \Cref{eq:consistency-equation}. To do so, we solve the linear system of equations
\begin{align*}
    \Theta &= a_2 \, \bar b(2,2),\\
     0 &= a_2  \,\bar b(1,2) + a_1  \,\bar b(1,1),\\
     0 &= a_2  \,\bar b(0,2) + a_1  \,\bar b(0,1) + a_0\Theta.
\end{align*}
\end{example}

\subsection{Approximating Terminal Curves for General Processes}
\label{subsec:general_processes}
The polynomial dynamics of the VIX process in the previous section allowed us to easily specify a terminal curve. Since the general model of \Cref{eq:general_vix} is generally neither polynomial nor time-homogeneous, we cannot apply the exact methodology. The conditional expectations are no longer polynomial functions, and hence we cannot eliminate the non-quadratic terms. Nevertheless, in this subsection, we utilize the insight from the polynomial processes to derive an approximate solution of a consistent terminal curve in the general model given by \Cref{eq:general_vix}. 

For the approximation, we also employ a time-constant polynomial coupling function $\psi_t = p \in \pol_m$ over the terminal interval. To derive the coefficients of $p$, we first project the conditional expectations, which are no longer necessarily polynomial functions, onto the polynomial subspace $\pol_m$ for some $m \geq 2$. 
Let $M_t^k(s; x)$ denote the conditional $k$-th moment of $\vix_s$ given $\vix_t = x$, and let $M_t^k(x)$ be its integral from $s = t$ to $s = t + \Theta$:
\[M_t^k(s; x) := \E\left[\vix^k_s \mid \vix_t = x \right], \quad \bar{M}_t^k(x) := \int_t^{t+\Theta} M_t^k(s; x) \ds.\]
We compute the integrated function $\bar{M}_t^k(x)$ at $t = T_F$, and then approximate the function with a polynomial $p_k \in \pol_m$, such that $p_k \approx \bar{M}_t^k(x)$. After projecting each moment $k \leq m$, we can then apply the same methodology as before and solve the linear system to determine $a_0,a_1,\dots,a_m$ of $p$ to obtain our consistent terminal curve.

It remains to define how to compute $M_t^k(s; x)$ or $\bar{M}_t^k(x)$ at $t = T_F$, and how to approximate the function with a polynomial. Since the $(\vix_t)_{t \geq t_0}$ is an It\^o process fulfilling linear growth conditions, we can utilize the Feynman-Kac theorem and solve the partial differential equation over $[0,\Theta]$:
\begin{equation}
\label{eq:fc}
(\partial_\tau + \mathcal{L}_\vix) M_t^k(t+\tau; x) = 0, \qquad M_t^k(t+\Theta; x) = x^k,
\end{equation}
with $\mathcal{L}_\vix M_t^k(s; x) = \mu(s,x) \partial_x  M_t^k(s; x) + \frac{1}{2} \sigma^2(s, x) \partial_{xx} M_t^k(s; x)$. This yields $ M_t^k(t + \tau; x)$ over the interval $ \tau \in [0,\Theta]$, after which we can integrate to obtain $\bar{M}_t^k(x)$. 

For the projection on the polynomial space to determine $p_k$, we solve the optimization problem by finding the closest $p \in \pol_k$ under the $L^2$-norm. Practically, it makes sense to apply a weighting function to the optimization problem, since we expect $\bar{M}_t^k(x)$ to have a leading $k$-th order term:
\begin{equation}
\label{eq:min}
    p_k := \arginf_{q \in \pol_k} \norm{ \frac{1}{x^k}\left({q - \bar{M}_t^k(.) }\right)}_2.
\end{equation}
Given the optimal polynomial, we extract all required coefficients to compute $a_0,a_1,\dots,a_m$  of the polynomial $p$. This yields our definition of the terminal curve $\bs{\psi}_{T_F}$ that fulfills Condition ii) and approximates Condition iii). For theoretical completeness we apply the $\epsilon$-approximation scheme from \Cref{lem:epsilon}
\subsection{Deriving the coupling function on $[t_0, T_F)$}
\label{subsec:backward}
After determining a consistent terminal curve, the model is completed by constructing $\psi_t$ at all earlier times $t \in [t_0,T_F)$. The main approach to propagating the terminal condition backward to $t_0$ relies on the backward relation \Cref{eq:backwards_main}, which requires evaluating the conditional expectation
\[ 
\E\left[ \psi_{t+\Theta}(\vix_{t+\Theta}) \,\middle|\, \vix_t=x \right],
\]
for a sufficiently fine time grid. 

For certain specifications of $(\vix_t)_{t\geq t_0}$, this conditional expectation is available in closed form. This occurs, for instance, when the VIX process is a martingale or belongs to a suitable class of polynomial processes, as in the examples above. When no analytical expression is available, the conditional expectation must be solved numerically. Although Least-Squares Monte Carlo (LSMC)~\cite{longstaff2001valuing} can estimate this expectation from simulated paths, it is often difficult to implement effectively. Selecting suitable basis functions to capture the nonlinear state dependence is challenging, particularly in high local-volatility regimes, and poor choices introduce projection errors that compound during backward stepping. 

Consequently, solving the Feynman-Kac backward PDE directly is the preferred methodology. We find $\phi_s(x) = \E\left[ \psi_{t+\Theta}(\vix_{t+\Theta}) \,\middle|\, \vix_s=x \right]$ on $s \in [t, t+\Theta]$ as a solution to the backward problem \begin{equation}
\label{eq:fc_for_psi}
(\partial_s + \mathcal{L}_\vix) \phi_s(x) = 0, \qquad  \phi_{t+\Theta}(x) = \psi_{t+\Theta}(x),
\end{equation}
with $\mathcal{L}_\vix\phi_s(x) = \mu(s,x) \partial_x\phi_s(x) + \frac{1}{2} \sigma^2(s, x) \partial_{xx} \phi_s(x)$. In practice, the backward propagation is evaluated over a discrete time grid
\[ 
\mathcal{T} = \{t_0, t_1, \ldots, t_N = T_F\}, 
\]
yielding a sequence of time-slice functions $\left\{ \psi_{t_n} \colon t_n \in \mathcal{T} \right\}$. To extend the coupling function continuously across the entire interval $[t_0, T_F]$, we perform linear interpolation between adjacent grid points. Specifically, for any $t \in [t_n, t_{n+1}]$, the coupling function is defined as
\[ 
\psi_t(x) = \frac{t_{n+1}-t}{t_{n+1}-t_n}\, \psi_{t_n}(x) + \frac{t-t_n}{t_{n+1}-t_n}\, \psi_{t_{n+1}}(x). 
\]
To verify the numerical accuracy of this continuous-time representation and assess consistency ex post, we study and compare the final samples of $\Theta \vix_t^2$ with numerical estimates of $\E\left[ \int_t^{t+\Theta} \psi_u(\vix_u)\,\mathrm{d}u \,\middle|\, \vix_t \right]$ over the relevant times and states using the following procedure:
\begin{prcd}[Non-parametric Consistency Analysis]
\label{prcd:npca}
The nonparametric consistency analysis analyzes the consistency in a pair of samples of $(\vix_t, v_t)_{t \geq t_0}$. The procedure estimates conditional means of realized volatility $v_t$ and compares the values to the current VIX $\vix_t$. For this procedure, we generate samples of $(\vix_t)_{t \geq t_0}$ on $[t_0, T_F + \Theta]$ and apply the coupling function to obtain, along each path, the corresponding samples of the instantaneous variance $(v_t)_{t \geq t_0}$. For a fixed evaluation time $t \in [t_0, T_F]$, we then record along each path the realized time-average
\begin{equation*}
    \bar{v}_t := \frac{1}{\Theta} \int_t^{t+\Theta} v_u \, \du,
\end{equation*}
approximated by trapezoidal quadrature on a uniform discretization of $[t, t+\Theta]$, together with the level $x := \vix_t$ observed at time $t$. We then partition the VIX paths into equally-spaced bins in log-space. Within each bin, we compute the empirical mean of the observed realized time-averaged variance
\begin{equation*}
    \widehat m(x_b) := \frac{1}{|\mathcal I_b|}\sum_{i \in \mathcal I_b} y_i,
\end{equation*}
where $\mathcal I_b$ denotes the set of path indices falling in bin $b$ and $x_b$ is the bin midpoint, as a nonparametric estimator of
\begin{equation*}
    \widehat m(x_b) \approx \frac{1}{\Theta}\E\left[\int_t^{t+\Theta} v_u \, \du \,\middle|\, \vix_t = x_b\right].
\end{equation*}
Under the consistency relation
\begin{equation*}
    x^2 = \frac{1}{\Theta}\E\left[\int_t^{t+\Theta} v_u \, \du \,\middle|\, \vix_t = x\right],
\end{equation*}
plotting $\widehat m(x_b)$ against $x_b^2$ across bins should then produce a scatter plot centered at the 45-degree line.
\end{prcd}
This concludes the construction of the coupling functions for general processes. We summarize the process in algorithmic form in \ref{app:algo}.
\subsection{Heston Model in \mname~Framework}
\label{sec:heston}
The \mname~model is in essence a stochastic volatility model, although the process $(v_t)_{t \geq t_0}$ is not modeled explicitly. To illustrate this fact, we express a well-established stochastic volatility model in the setting of the \mname~framework. The Heston model has the following $\Q$-dynamics for the instantaneous variance process:
\begin{equation}
    \label{eq:Heston}
    \dv_t = \kappa (\theta - v_t) + \xi \sqrt{v_t} \dW.
\end{equation}
Its first conditional moment is given analytically by
\[\E[v_s | v_t ] = \theta + (v_t - \theta) \e^{-\kappa (s-t)}, \quad t<s.\] 
For that reason, the VIX process $(\vix_t)_{t\geq t_0}$ can be derived explicitly as a function of $v_t$. Integrating the right-hand side of the consistency equation (\ref{eq:consistency-equation}), and setting it equal to the left-hand side, we find:
\[  \Theta\vix_t^2 = \mathbb{E}\!\left[\int_t^{t+\Theta} v_u\,\du \;\middle|\; \mathcal{F}_t\right] = \Theta^2(1-B) + \Theta B v_t,\]
with $B = \frac{1- \e^{-\Theta \kappa}}{\Theta \kappa}$.  Solving for $v_t$, we get an explicit expression for a coupling function.
\[v_t = \psi_t(X_t), \quad \psi_t(x) = \frac{x^2}{B} + \Theta(1 - \frac{1}{B})\]
Using It\^os lemma on the inverse map of $\psi_t$, we can then derive the explicit dynamics for $(\vix_t)_{t \geq t_0}$ that lead to the Heston model under the \mname~framework. We find
\[
\d\vix_t = \mu(\vix_t) \, \dt + \sigma(\vix_t) \, \dW_t,
\]
where the drift $\mu(\vix_t)$ and diffusion $\sigma(\vix_t)$ are given by:
\[
\mu(\vix_t) = \frac{B\kappa\theta - \kappa\left(\vix_t^2 + \theta(1 - B)\right)}{2\vix_t} - \frac{\xi^2\left(\vix_t^2 - \theta(1 - B)\right)}{8\vix_t^3}, \quad
\sigma(\vix_t) = \frac{\xi}{2\vix_t} \sqrt{\vix_t^2 - \theta(1 - B)}.
\]
In the \mname~model, we normally do not know the coupling function $\psi_t$ \emph{a prior}, since we first specify the dynamics of the VIX process. In the Heston case, we have the coupling function explicitly. Since the Heston-implied VIX process is not polynomial, we can apply the projection scheme to replicate the coupling function and to assess the numerical accuracy of the method. We compute the conditional first and second moments of $(\vix_t)_{t \geq t_0}$ by solving the Feynman-Kac backward equation of \Cref{eq:fc}, and fit the corresponding polynomials. Then, we find the coupling function $\psi_t$ as $p \in \pol_2$ with $p = a_2x^2+a1x+a_0$. For an example of $\kappa = 2, \theta = 0.04$ and $\xi = 0.5$, we find the following values:
\begin{table}[H]
  \centering
  \begin{tabular}{lccc}
    \hline
& $a_2$ & $a_1$ & $a_0$ \\
    \hline
    \textbf{Approximation} & $1.085647$ & $-4.11 \times 10^{-8}$ & $-0.003426$ \\
    \textbf{Exact Values}  & $1.085647$ & $0.000000$             & $-0.007137$ \\
    \hline
  \end{tabular}
  \caption{Comparison of polynomial coefficients in Heston VDV for $\kappa = 2, \theta = 0.04$ and $\xi = 0.5$.}
  \label{tab:coefficients}
\end{table}
The example shows that the polynomial approximation method is able to approximate the true coupling function well.
\section{Local-Volatility-VIX}
\label{sec:lv-vix}
In this section we present the \emph{Local-Volatility-VIX (LV-VIX)} model, an implementation of the \mname~model. In this model, the VIX process follows a \emph{mean-reverting local volatility process}, as defined in~\cite{drimus2013local}. Let $k > 0$ be fixed and let $\theta \colon [t_0,\infty) \to \R_+$ be a positive map. We consider the following model dynamics:
\begin{align}
    \frac{\dS_t}{S_t} &= \sqrt{v_t}\, \dW^{(S)}_t,\label{eq:main_S}\\
    v_t &= \psi_t(\vix_t),\label{eq:main_v}\\
    \d\vix_t &= k \,( \theta_t - \vix_t)\,\dt +  \sigma_{LV}(t,\vix_t)\vix_t\,\dW_t^{(\vix)}.\label{eq:main_vix}
\end{align}
The mean-reverting process $(\vix_t)_{t\geq t_0}$ has parameters $k$, defining the mean-reversion speed, and a long-term mean of $\vix_t$ defined by $\theta_t$. Since $\theta_t$ is a time-dependent function, we use it to calibrate to the VIX future curve (see \ref{app:lvol}). The diffusion function $\sigma_{LV}(t,k)$ is the \emph{local volatility function} of the VIX process, derived following the derivation of Dupire's formula~\cite{Oosterlee_Grzelak_2020}. Since $(\vix_t)_{t\geq t_0}$ is not a discounted martingale, but has a mean-reversion component for the drift, we require an adjustment of Dupire's formula. The derivation of the formula is outlined in~\cite{drimus2013local}. We find
\begin{equation}\label{eq:Dupire-VIX}
\sigma_{LV}^2(T,K)  =\frac{ k(\theta_T - K) \cdot \frac{\partial C}{\partial K} (T,K) + (r + k) \cdot C(T,K) + \frac{\partial C}{\partial T} (T,K)}{\frac{1}{2} K^2 \cdot \frac{\partial^2 C}{\partial K^2} (T,K)},
\end{equation}
where $r$ is the discounting rate and $C(T,K)$ is the price of a call option on the VIX index with strike $K$ and expiry $T$:
\begin{equation}
    C(T,K) :=  \E\left[\e^{-rT}(\vix_T - K)^+\right].
\end{equation}

A key feature of the outlined model is that every parameter is used to calibrate at most one type of market data, making the parameter impact clear and simple. \Cref{tbl:params} lists the target market data for each parameter:
\begin{table}[H] 
\centering 

\begin{tabular}{lll} \hline 
Market Data & Parameters & Type\\
\hline \hline 
VIX Futures & $\theta_t$ & Map $\left([t_0,\infty) \to \R_+ \right)$ \\ 
VIX Options & $\sigma_{\mathrm{LV}}$ & Map $ \left([t_0,\infty) \times \R_+ \to \R_+ \right)$ \\ 
SPX Options & $k, \rho$ & Scalar \\ \hline \end{tabular}
\caption{Market Data and Parameter Overview} 
\label{tbl:params}\end{table}

Using the mean-reverting dynamics $\mu(t,x) = k(\theta_t - x)$ and local volatility $\sigma(t,x) = x \cdot \sigma_{LV}(t,x)$ for the VIX and plugging in the functions, we therefore find the backward equation for $v_t$ over a timestep of $\Theta$:
\begin{equation}
\label{eq:backwards_vix}
    \psi_t(\vix_t) = \E\left[\psi_{t+\Theta}(\vix_{t+\Theta}) \mid \vix_t \right]  - \Theta \left[2\, k\, \vix_t\,(\theta_t - \vix_t) + \sigma^2_{LV}(t,\vix_t)\,\vix_t^2\right], \quad \text{a.s.}
\end{equation}

The local volatility process $(\vix_t)_{t\geq t_0}$ is generally not a polynomial process. For this reason, the polynomial approximation scheme as described in \Cref{subsec:general_processes} is required to estimate a consistent terminal curve $\bs{\psi}_{T_F}$. Then, the coupling function on $t \in [t_0,T_F)$ can be propagated through the backward equation. This requires the computation of $\E[\psi_{t+\Theta}(X_t) \mid X_t]$ by solving the Feynman-Kac backward PDE of [pde]. The positivity of the function $\psi$ can be checked post-hoc and floored should it be required. \Cref{fig:coupling_function} is a visualization of a coupling function derived from an LV-VIX model.
\begin{figure}[H]
    \centering
    \includegraphics[width=0.7\linewidth]{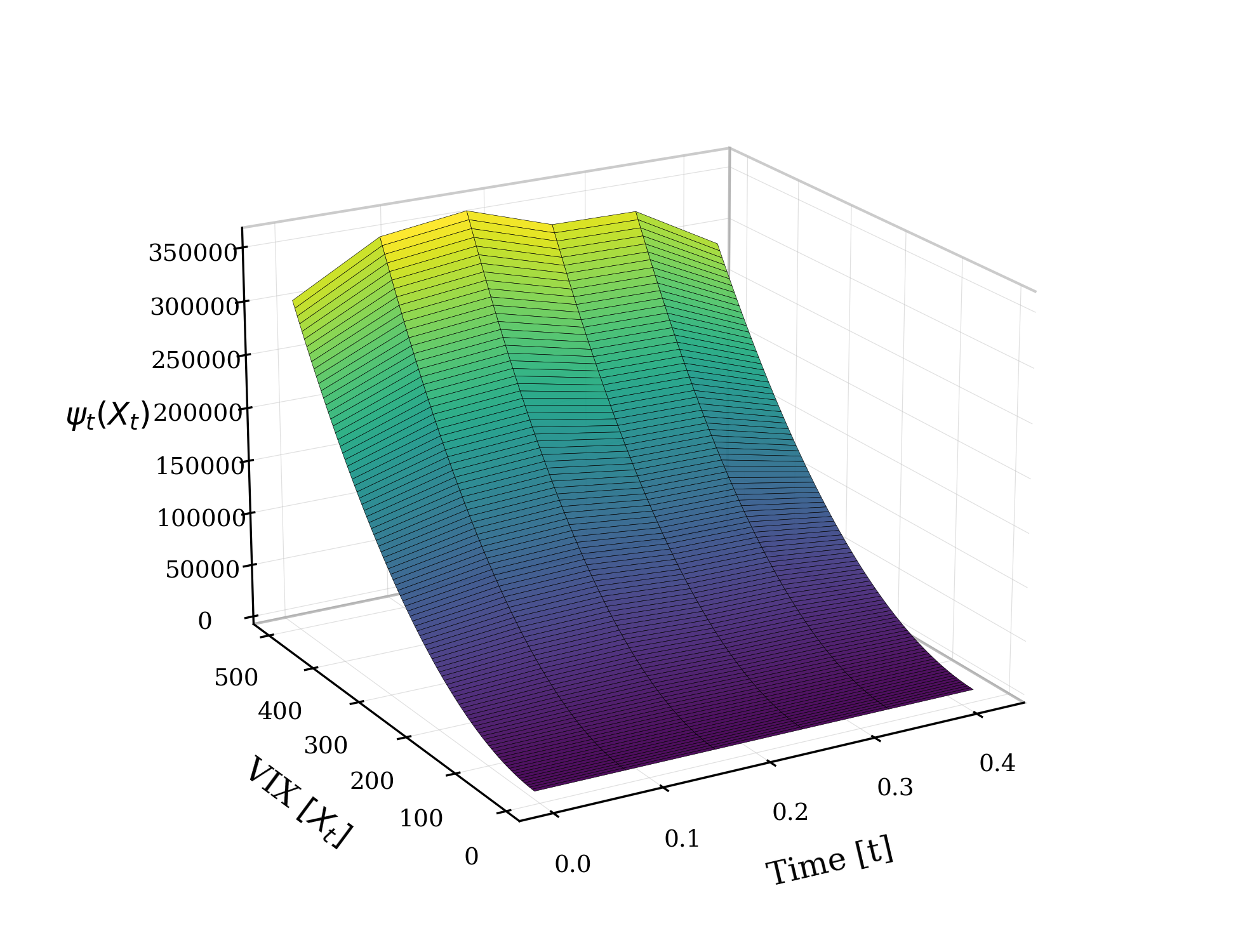}
    \caption{Visualization of a coupling function}
    \label{fig:coupling_function}
\end{figure}
\subsection{Model Parameter Study}
The LV-VIX model has free model parameters $k,\theta_t, \rho$ that need to be specified. Since the future curve of the VIX defines the drift term of the local vol model, we calibrate $\theta_t$ as a function of $k$ and the VIX future curve (see \ref{app:lvol}), according to \Cref{tbl:params}. The remaining parameters $k,\rho$ are used to calibrate the model SPX options to the SPX market. 

A particularly useful feature of the model is that the parameters used for the SPX fitting do not alter the marginal distributional properties of $(\vix_t)_{t\geq t_0}$, and hence can be calibrated independently. 
For instance, the temporal dependence structure of volatility can be controlled through the mean-reversion parameter $k$, while preserving the marginal distributions of the process. More precisely, whenever $k$ is changed, the local-volatility function $\sigma_{LV}$ is recalibrated so that the distribution of $\vix_t$ for every fixed time $t \geq t_0$ remains unchanged. Consequently, modifying $k$ does not affect the model-implied cross-sectional distributions and therefore does not alter the corresponding European option prices. 

While the one-dimensional marginal distribution is preserved, the joint distribution of $\vix_t$ across different times is not. In particular, path-dependent quantities, such as the temporal covariance $\operatorname{Cov}\!\left(\vix_t,\vix_s\right)$ for times $s \neq t$ and, more generally, the persistence of volatility shocks, depend on the choice of $k$. The parameter $k$ therefore provides a direct mechanism for calibrating the autocorrelation structure of volatility independently of its marginal distributions. This distinction is illustrated in \Cref{fig:varying_k_paths}, where we compare simulated trajectories obtained with $k=1$ and $k=12$.
\begin{figure}[H]
    \centering
    \includegraphics[width=\linewidth]{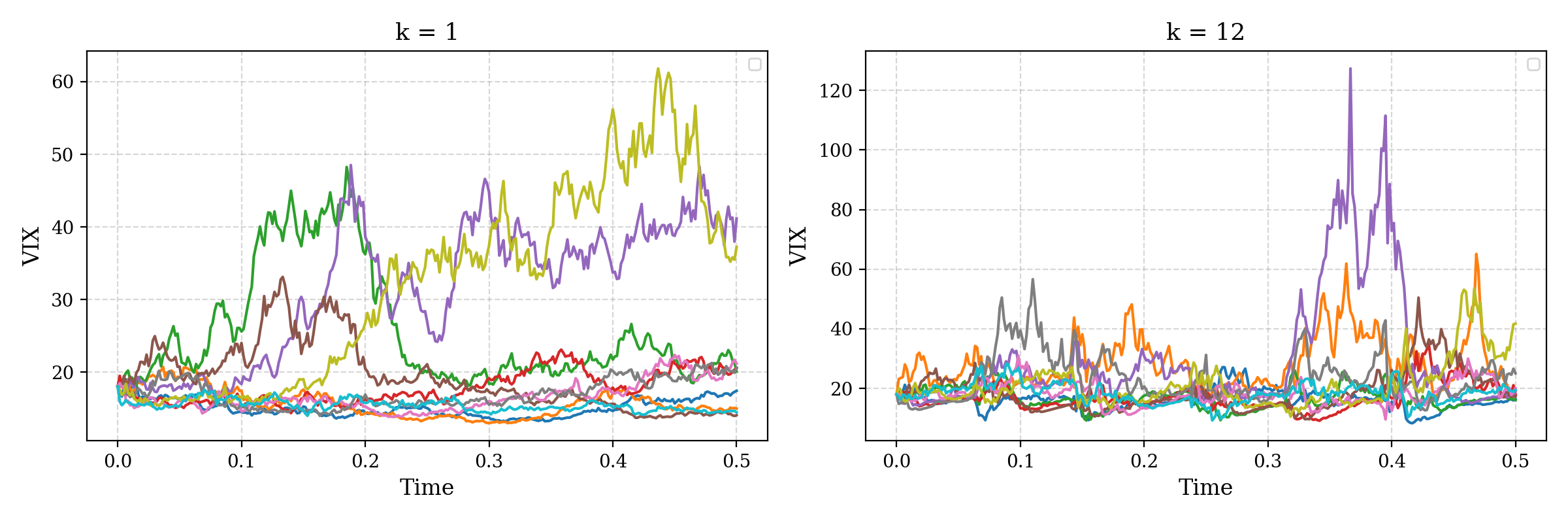}
    \caption{VIX Sample paths for two values of $k$, calibrated to same marginal distributions: Left-hand side: $k = 1$, Right-hand side: $k = 12$.}
    \label{fig:varying_k_paths}
\end{figure}
For $k=1$, the relatively weak mean reversion allows deviations from the long-run mean to persist over longer periods, resulting in smoother and more prolonged excursions. By contrast, for $k=12$, the stronger mean-reverting force pulls the process back toward its long-run mean much more rapidly. At the same time, because the recalibration of $\sigma_{LV}$ preserves the marginal distribution at each time, the same range of volatility levels must still be attained. The simulated paths consequently display more frequent, but less persistent, deviations from the long-run mean. Thus, the parameter $k$ changes how volatility levels are reached and how long they persist, rather than changing the set of marginal distributions generated by the model. \Cref{fig:lv_var_k} shows how the local volatility changes when the parameter $k$ is varied.
\begin{figure}[H]
    \centering
    \includegraphics[width=\linewidth]{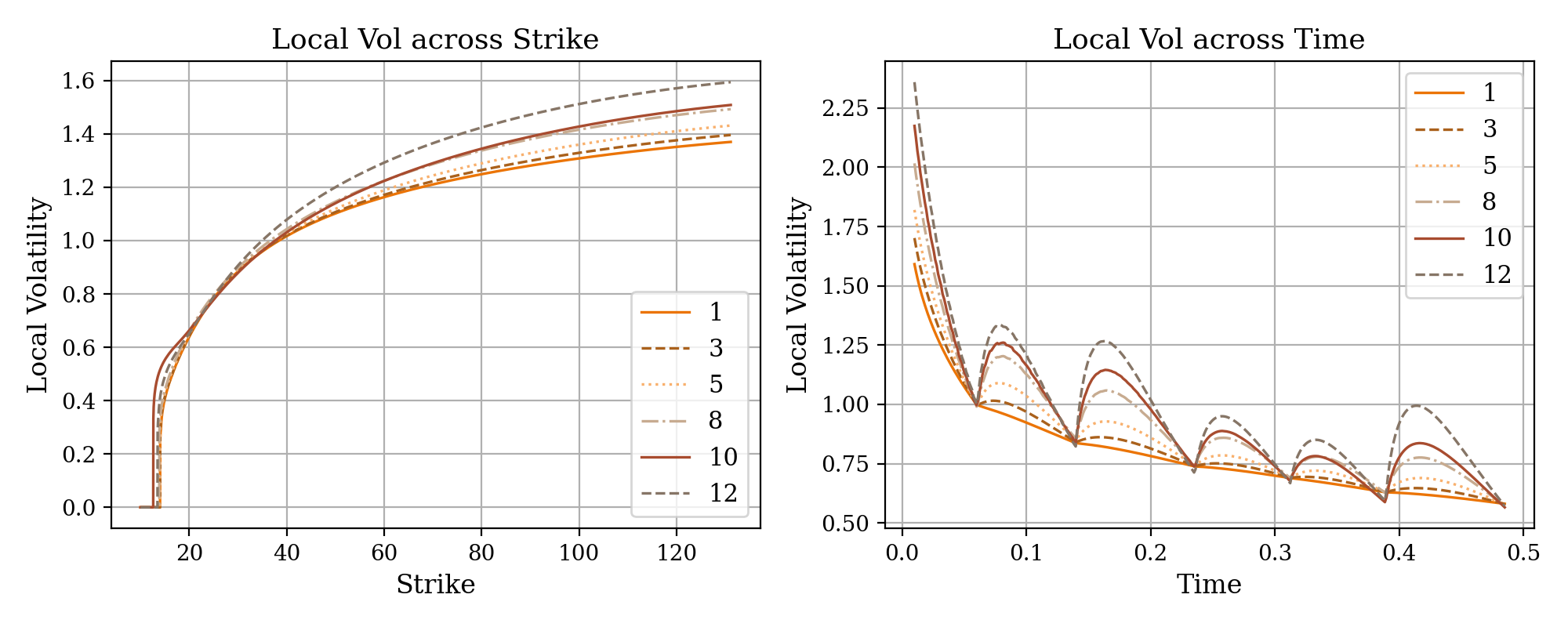}
    \caption{Local vol functions $\sigma_{LV}$ across time and space. Parameter $k\in \{1,3,5,8,10,12\}$ }
    \label{fig:lv_var_k}
\end{figure}
Through the temporal volatility dependency of $k$, and the correlation coefficient $\rho$ that governs the correlation between the stochastic driver of the SPX and the VIX, the LV-VIX model has two parameters that can be used to fit the SPX option market. \Cref{fig:iv_var_k} shows how each parameter affects the SPX implied volatilities.
\begin{figure}[H]
    \centering
    \includegraphics[width=\linewidth]{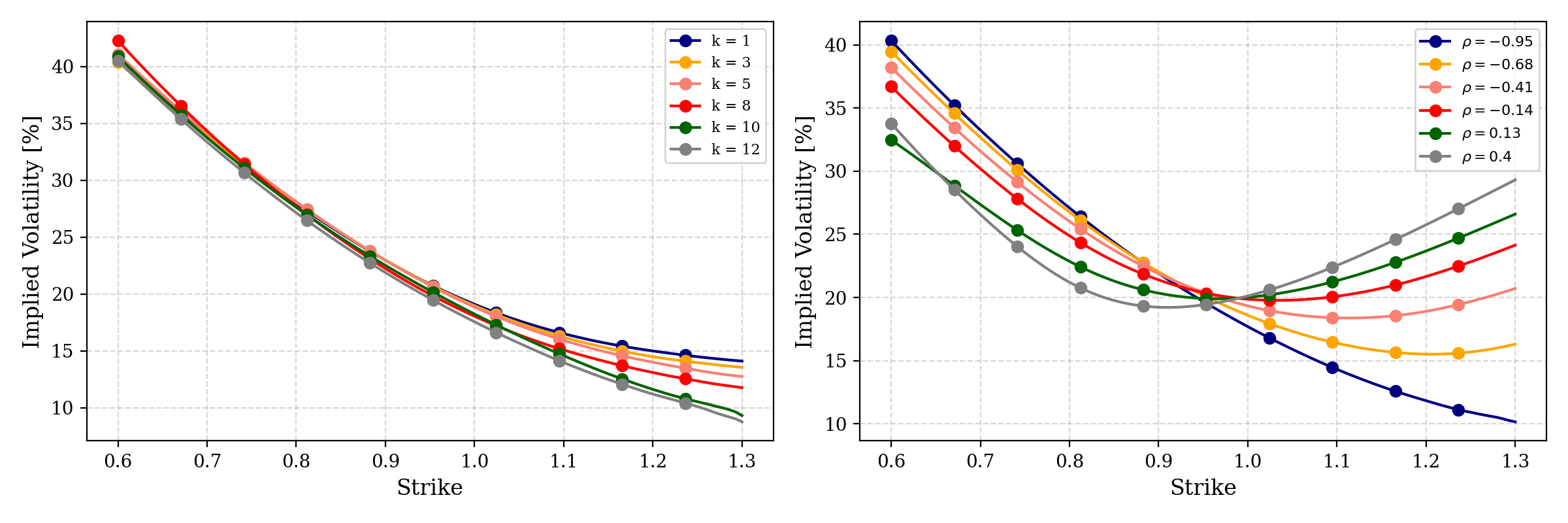}
    \caption{SPX Implied Volatility Shapes for varying $k$ and $\rho$: Left-hand side: $\rho = -0.95$, Right-hand side: $k = 12$ }
    \label{fig:iv_var_k}
\end{figure}
\section{Numerical Experiments}
 \label{sec:num}
 In this section, we apply the LV-VIX model from \Cref{sec:lv-vix} to real market data of the SPX and VIX. We briefly discuss the data processing and cleaning aspects, and then show how the model was able to fit the markets.
 \subsection{Data \& Data Processing}
 The option market bid and ask quotes for the VIX and SPX are retrieved from the provider Thetadata~\cite{ThetaData2026}, and cleaned to exclude illiquid or arbitrage quotes. The snapshot time is the U.S. markets EOD (4:30 PM ET) on April 27th, 2026. We consider strikes that are quoted at a positive bid and positive ask price to ensure only relevant quotes are used for the fit. Put option prices are converted to call option prices using the put-call parity. \Cref{tbl:summary_stats} summarizes the number of quotes and maximum bid-ask (implied volatility) spread $\frac{\text{Ask} - \text{Bid}}{\text{Mid}}$.
 \begin{table}[H]
\centering
\caption{Summary of Option Quotes and Maximum Relative Spreads}
\label{tbl:summary_stats}
\begin{tabular}{lccccc}
\hline
\multicolumn{6}{c}{\textbf{VIX Options (Spot 18.02)}} \\
\hline
\textbf{Maturity} & \textbf{Forward} & \textbf{Min Strike} & \textbf{Max Strike} & \textbf{Quotes} & \textbf{Max \% Spread} \\ \hline
2026-05-19 &20.05 &17 & 55.0 & 37 & 18.92\% \\
2026-06-17 &21.09& 16 & 85.0 & 45 & 20.80\% \\
2026-07-22 & 21.95&16 & 95.0 & 38 & 16.81\% \\
2026-08-19 & 22.20&16 & 95.5 & 38 & 19.87\% \\
2026-09-16 & 22.57&16 & 95.5 & 38 & 21.90\% \\
2026-10-21 & 23&16 & 95.0 & 38 & 19.18\% \\ \hline
\end{tabular}\\
\begin{tabular}{lccccc}
\hline
\multicolumn{6}{c}{\textbf{SPX Options (Spot 7152.72)}} \\
\hline
\textbf{Maturity} & \textbf{Forward} & \textbf{Min Strike} & \textbf{Max Strike} & \textbf{Quotes} & \textbf{Max \% Spread} \\ \hline
2026-05-15 & 7184.96 & 5250& 7640& 410 & 8.82\% \\
2026-06-18 & 7205.15 & 3400 & 8200& 494 & 6.82\% \\
2026-07-17 & 7222.97 & 1800 & 8700& 468 & 6.33\% \\
2026-08-21 & 7244.04 & 1200& 9200& 412 & 6.24\% \\
2026-09-18 & 7259.87 & 1400 & 9800& 367 &6.19\% \\
2026-10-16 & 7278.60 & 1400& 10200 & 304 & 6.06\% \\
\end{tabular}
\end{table}
 For the VIX futures data, we retrieve the daily settlement prices, as provided by CBOE~\cite{CBOEVIXFUTURES}. The settlement prices roughly correspond to EOD market prices. For any expiry date not listed, we use a linear interpolation method to determine the corresponding forward price. 
\subsection{Local Volatility for VIX}
In the first part of the model calibration, we calibrate the local volatility process $(\vix_t)_{t \geq t_0}$. We employ an adaptation of Drimus and Farkas' method~\cite{drimus2013local}, which we outlined in \ref{app:lvol}. The forward curve consists of $M=6$ maturities $t_1,t_2,\dots,t_6$, for which the forward prices $\{F_{X, t_m} \colon  0\leq m \leq M\}$ are listed in \Cref{tbl:summary_stats}. We construct the forward curve $F_X \colon [t_0,\infty )$ such that $F_X(t_0)$ is the spot price and $F_{X,m} = F_X(t_m)$. For the parameter $k$ we use $k=12$. The resulting mean-reversion parameters are then derived as $\theta_1, \theta_2, \dots,\theta_5  = 23.33,22.28,22.68,22.5,23,23.37$. \Cref{fig:future_curve} shows the corresponding curve and parameters for $\theta_m$ for $k=12$ as well as $k \in \{5,8\}$ as comparison.
\begin{figure}[H]
    \centering
    \includegraphics[width=\linewidth]{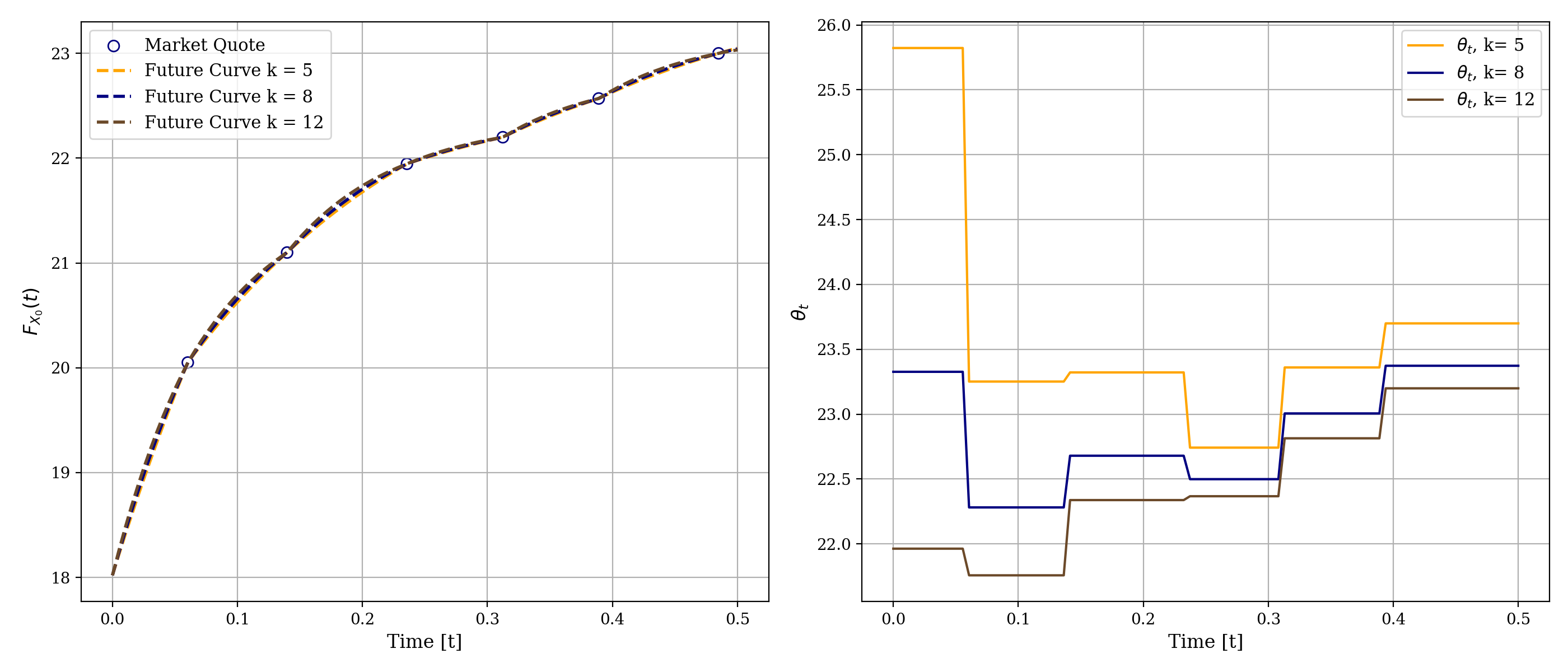}
    \caption{Calibrated VIX Future curve and corresponding $\theta_t$ on April 27, 2026. Parameter $k\in \{5,8,12\}$ }
    \label{fig:future_curve}
\end{figure}
The interpolation functions $\nu_m(.)$ (see \ref{app:lvol}, \cite{drimus2013local}) are fitted on mid market implied volatilities. Contrary to Drimus and Farkas, we do not use a parametric form of $\nu_m(.)$, but a step-function approach as described in \ref{app:lvol}. This yields a twice-differentiable option price surface for $t \in [t_0,t_6]$ and any desired strike $K$. From the price surface we can derive the implied volatility surface and the derivative of the implied volatility surface to compute \Cref{eq:app_lv-vol}, which is equivalent to \Cref{eq:Dupire-VIX}. 
The resulting local volatility model after calibration replicates the smoothed VIX options market by design. We simulate Monte Carlo paths and price model options to compare to the market quotes. \Cref{fig:vix_fit} shows that the model fit of the options lies largely within the bid-ask spread of the VIX options.
\begin{figure}[H]
    \centering
    \includegraphics[width=\linewidth]{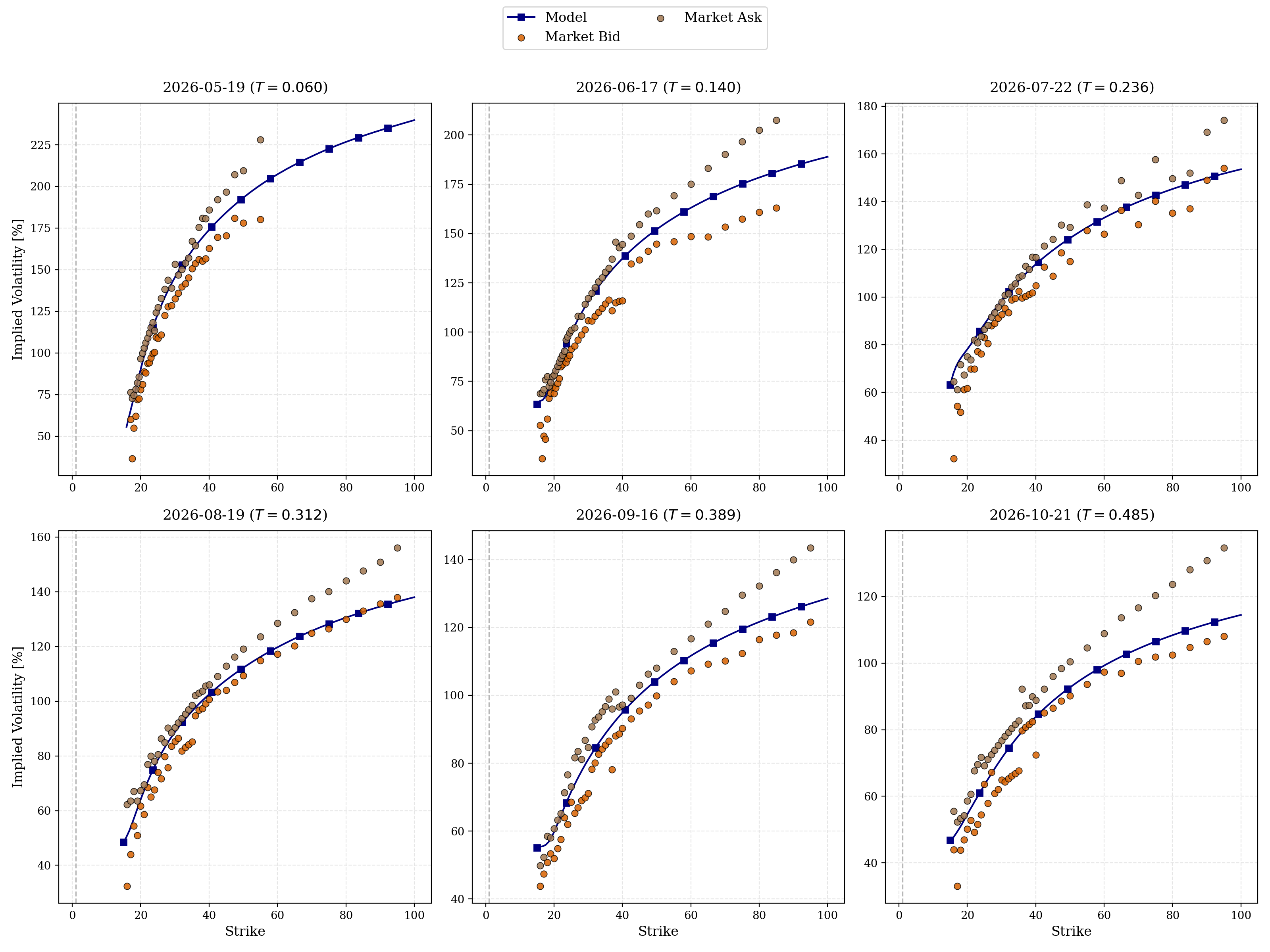}
    \caption{VIX Options Model to Market Fit.}
    \label{fig:vix_fit}
\end{figure}
\subsection{Instantaneous Variance of the SPX}
In the final step of the model calibration, we derive the coupling function $\psi_t$ on the interval $[t_0,T_F+\Theta]$ and then fit the final parameter $\rho$ to fit the SPX option skew to the market data. Since the calibration uses maturities up to October 2026, we choose $[t_0, T_F+\Theta] = [0,0.5]$. 

For the terminal curve, we derive a second-order polynomial $p \in \pol_2$ from the numerically computed values for $\bar{M}_{T_F}^1$ and $\bar{M}_{T_F}^2$. We find $p(x) = a_0 + a_1 x + a_2 x^2$ with \[a_2 \approx 1.16, a_1 \approx -4.19, a_0 \approx -101.37.\] We then derive the coupling function $\psi$ on the entire interval $[t_0,T_F+\Theta]$ with the backward equation (\ref{eq:backwards_main}), repeatedly solving the conditional expectation using a numerical PDE solver. To empirically validate the consistency on the entire interval, we analyze the scatter plot for the nonparametric consistency analysis of \Cref{prcd:npca}. \Cref{fig:npca} shows that the coupling function produces consistency across two arbitrarily selected times.
\begin{figure}[H]
    \centering
    \includegraphics[width=0.7\linewidth]{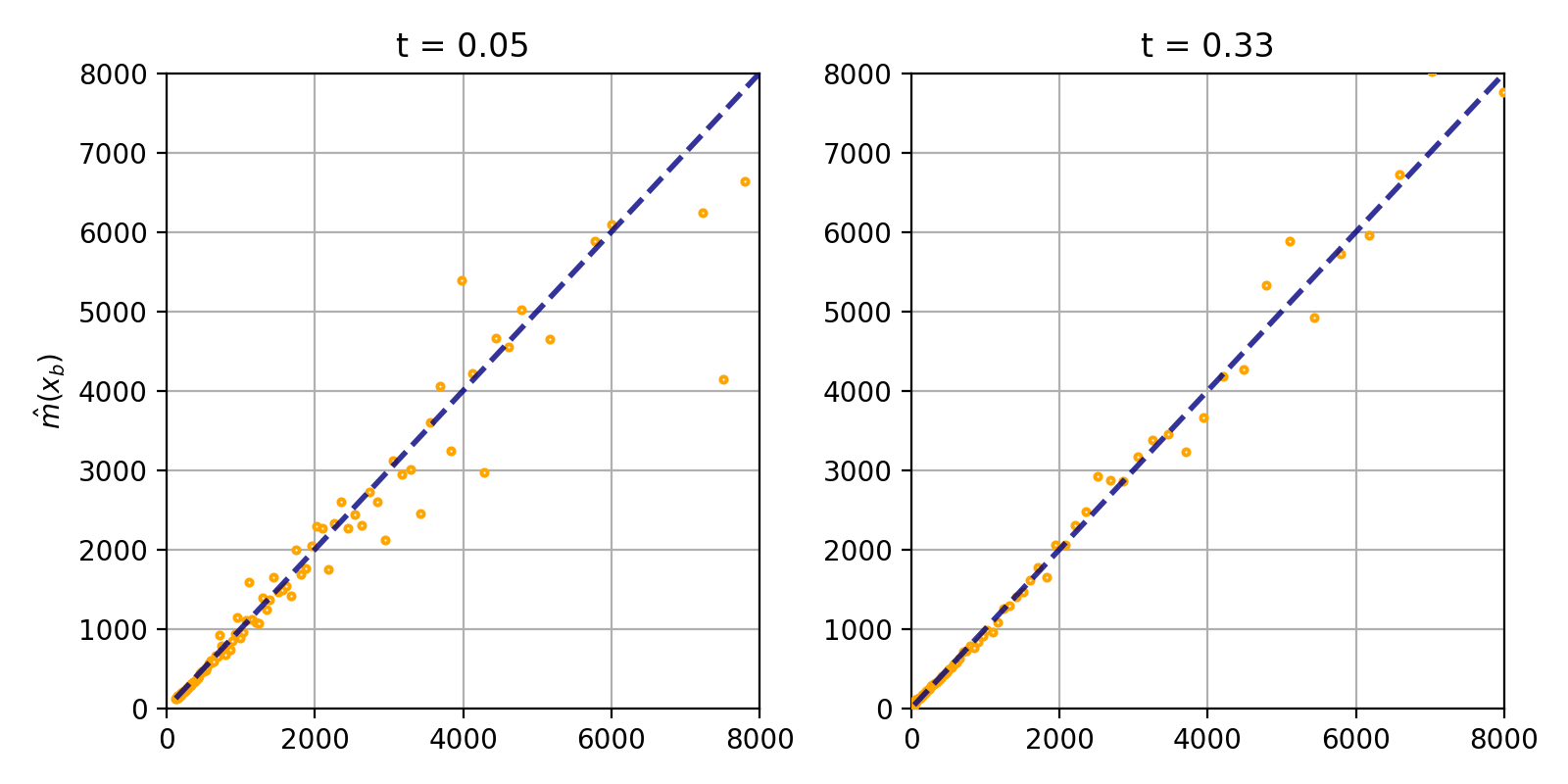}
    \caption{Nonparametric Consistency Analysis: The VIX and the instantaneous variance are consistently calibrated.}
    \label{fig:npca}
\end{figure}
Finally, we compute SPX samples from the instantaneous volatility paths and compute SPX option prices. We calibrate the correlation coefficient to $\rho=-0.95$. The fit of the model to the market implied volatilities is shown in \Cref{fig:spx fit}. The model price is largely in line with the market bid-ask prices observed at the time, proofing the \mname~model as valid joint SPX/VIX framework.
\begin{figure}[H]
    \centering
    \includegraphics[width=\linewidth]{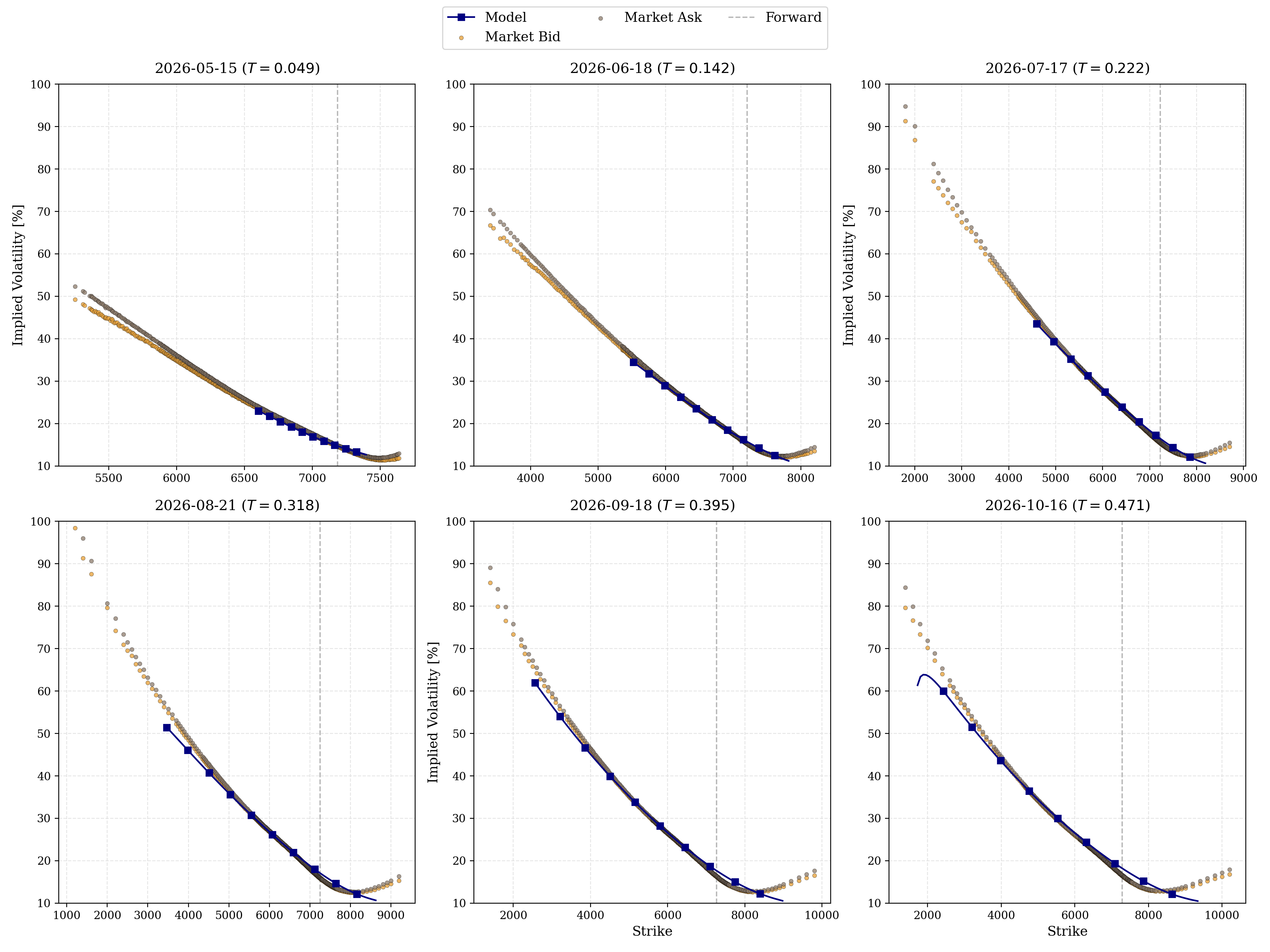}
    \caption{SPX Options Model to Market Fit.}
    \label{fig:spx fit}
\end{figure}
We note that the largest discrepancy between market and model occurs for out-the-money call options. A possible explanation is the simplification of a linear correlation structure between SPX and VIX drivers. As existing research suggests~\cite{papanicolaou2014regime}, the correlation between the VIX and SPX is more negative for negative SPX returns than for positive SPX returns. The linear correlation coefficient $\rho$ does not capture this intricate feature, hence overpricing out-the-money call options.
\subsection{Computational Times}
Finally, we briefly comment on the computational effort of calibrating the local volatility function and deriving the coupling function. The numerical efforts requiring heavier computations can be classified into 4 stages of the calibration. \Cref{tab:calibration-time} shows the 4 stages and the corresponding approximated required times.
\begin{table}[H]
    \centering
    \begin{tabular}{llll}
        \toprule
        \textbf{Stage} & \textbf{Problem} & \textbf{Timing} & \textbf{Equation} \\
        \midrule
        Calibration of $\nu_i$              & Linear system                    & $<$1s    & \Cref{eq:linearsystem} \\
        Calibration of local volatility    & Linear system                    & $<$1min  & \Cref{eq:linearsystem} \\
        Calibration of terminal curve      & Partial differential equation  & $<$30s   & \Cref{eq:fc} \\
        Calibration of coupling function   & Partial differential equation  & $<$1min  & \Cref{eq:fc_for_psi} \\
        \bottomrule
    \end{tabular}
    \caption{Computational cost of each calibration stage.}
    \label{tab:calibration-time}
\end{table}
The main computational effort for the coupling function concerns solving equations (\ref{eq:fc}) and (\ref{eq:fc_for_psi}). The speed of these equations depends heavily on the smoothness of the diffusion function $\sigma_{LV}$, since a noisy diffusion function requires a dense grid to solve it accurately. For the derivation of the coupling function in \Cref{fig:coupling_function}, we solved a total of 7 Feynman-Kac problems, each over a time-span of $\Theta$.
\section{Conclusions}
A joint SPX-VIX framework is a stochastic model for the S\&P 500 that is able to price SPX options, VIX options, and VIX futures according to the market. Defining and calibrating such a joint model requires high modeling flexibility to produce realistic simulations of both the volatility and price processes of the SPX. The high dimensionality of the available market data stemming from three separate markets makes parametric model calibration extremely difficult, forcing modelers to opt for nonparametric alternatives.

In this paper, we introduced the \mname~model, a stochastic volatility model for the SPX, where the volatility process emerges as a latent factor of an explicitly modeled VIX process. The two processes of the VIX and SPX are coupled with a coupling function, which expresses the instantaneous variance of the SPX as a function of the VIX. We derived an analytically consistent coupling function for the case when the VIX process $(\vix_t)_{t \geq t_0}$ follows a geometric Brownian motion, and more generally for polynomial processes. We then show how to derive an approximate consistent coupling function for the general model, where $(\vix_t)_{t \geq t_0}$ is a time-inhomogeneous non-polynomial process. 
A key feature of the \mname~model is that it has distinct calibration possibilities for each relevant market, separating the calibration process into three steps and avoiding tedious parameter searches.

In \Cref{sec:num} we introduced the Local Volatility-VIX model, an implementation of a VDV model capable of producing realistic VIX patterns. In this model, the VIX process follows a mean-reverting local volatility process that allows calibration to both VIX futures and VIX options following the work of Drimus and Farkas~\cite{drimus2013local}. We then applied the Local-Volatility VIX model to market data from April 2026. The real-data experiment showed three achievements: Firstly, the VIX process $(\vix_t)_{t \geq t_0}$ is flexible enough to be calibrated to VIX futures and VIX options (\Cref{fig:future_curve,fig:vix_fit}). Secondly, the derived instantaneous SPX variance $(v_t)_{t \geq t_0}$ through the coupling function is accurate enough to provide consistency between the SPX and the VIX across the entire time $[t_0,T_F+\Theta]$ (\Cref{fig:npca}). Lastly, by setting the final model parameter $\rho$, the model can be calibrated to the SPX option market, providing a good model fit across all maturities (\Cref{fig:spx fit}). The experiments show that a realistic stochastic model for the VIX, derived from the VIX option and futures market can explain a substantial part of the volatility dynamics of the SPX. 
\label{sec:conclusion}
\bibliographystyle{abbrvnat}
\bibliography{bib}

\appendix
\section{Local Volatility with Mean-Reversion}
\label{app:lvol}
The local volatility model is a well-known stochastic model for pricing derivatives. In its canonical form, the model describes a process $(S_t)_{t>t_0}$ given
\[\dS_t = r S_t \dt + \sigma_{LV}(t,S_t) S_t\dW_t,\]
where $r$ is a fixed rate describing the constant drift to the forward price of $S_t$. The model is easily calibrated to any options market by a closed-form of $\sigma_{LV}$. It is assumed by the model that $S_t$ is a discounted martingale under the risk-neutral measure. Under the fundamental theorem of asset pricing, this is a no-arbitrage condition for any tradable asset. The VIX process $(\vix_t)_{t\geq t_0}$, as defined in \Cref{eq:main_vix} is not a tradable asset, as there is no market to buy or sell the VIX spot price. Instead, the VIX is traded with future contracts. This future curve describes the expected future spot price of the VIX under the risk-neutral measure. For this reason, $(\vix_t)_{t\geq t_0}$ is generally not a discounted martingale, violating a key requirement of the canonical local volatility model. Nevertheless, we can derive a similar local volatility framework for non-martingale processes when the drift term fulfills certain properties, as shown by Drimus and Farkas~\cite{drimus2013local}. The authors introduce a model with mean-reverting drift $ k( \theta_t - \vix_t) \dt$ for the VIX, and show how it can be calibrated to both the VIX future and options market. As this framework is central to our work, we include a brief recap of it here. Here we will assume that $(\vix_t)_{t\geq t_0 }$ follows the stochastic differential equation
\[\d\vix_t = k( \theta_t - \vix_t) \dt + \sigma_{LV}(t,\vix_t) \vix_t \dW_t,\]
for a Brownian motion $W_t$ and parameters as described in \Cref{sec:model}.
\subsection{The future curve}
A key feature of the process $(\vix_t)$ is that it allows a global fit to the future curve. The future, defined as $F_{\vix_t}(T):= \E[X_T | \F_t] $, is the mean of the future spot, given the current information.  Using It\^o's lemma on $\e^{k t}\cdot \vix_t$, we find
\[ \e^{k(T-t)} X_T = \vix_t + k\int_{t}^T \e^{k(u-t)} \theta_u \du + \int_{t}^T \e^{k(u-t)} \sigma_{LV}(u,\vix_u) \vix_u \dW_u.\]
Taking the conditional expectation with respect to $F_t$, we then find
\[F_{\vix_t}(T) = \vix_t \e^{-k(T-t)} +  k\e^{-k(T-t)}\int_{t}^T \e^{k(u-t)} \theta_u \du.\]
Suppose now that for time $t = 0$, we obtain a grid of $N$ future points $\mathcal{FC} = \{F_{X_0}(T_1),F_{X_0}(T_2), \dots, F_{X_0}(T_N)\}$. If we fix the parameter $k > 0$, we aim to determine the map $\theta_t$, such that the model calibrates to $\mathcal{FC}$. We look for a piecewise constant map for $\theta_t$ with step points at $T_1,T_2,\dots,T_N$. The first point $F_{X_0}(T_1)$ is found by assuming $\theta_t$ to be constant on $[0,T_1]$, after which we obtain:
\[\theta_{t} = \frac{F_{X_0}(T_1) \e^{kT_1} - X_0}{\e^{kT_1} - 1} , \quad t \in [0,T_1].\]
    For any $n \leq N$, we denote the constant map between $(T_{n-1}, T_n]$ as $\theta_n$. Given the map on $[0,T_1]$, we can then determine the map $\theta_2$ on $(T_1,T_2]$ given by $F_{X_0}(T_2)$. This procedure is repeated for all $n \leq N$, for which we then obtain:
\[\theta_{n} = \frac{F_{X_0}(T_n)\e^{kT_n} - X_0 - \sum_{i=1}^{n-1} \theta_i\left(\e^{kT_i} - \e^{kT_{i-1}}\right)}{\e^{kT_n} - \e^{kT_{n-1}}} , \quad t \in [T_{n-1},T_n].\]
In \Cref{fig:future_curve} we show the calibrated future curve and its corresponding $\theta$-function on April 27th 2026 for $k  \in \{5,8,12\}$. The plot shows that $k$ affects the curvature of the future curve between the market quotes.

\subsection{The Local Volatility Function}
Having determined the parameter $k$ and its corresponding step-function $\theta_t$, we calibrate process $(\vix_t)_{t\geq t_0}$ to its options market. The derivation of the local volatility function $\sigma_{LV}$ for the mean-reversion model follows the same arguments as the derivation of the canonical local volatility model for discounted martingales from the definition of the call option price as an integral over the risk-neutral density. The key requirement for the mean-reversion term is that it is linear in the state variable $\vix_t$, such that its second-order derivative vanishes. We refer to~\cite{Oosterlee_Grzelak_2020} for a derivation of the canonical model, and to~\cite{drimus2013local} for a derivation of the mean-reverting adjustment. The function is given by
\begin{equation}
\label{eq:app_lv}
\sigma_{LV}^2(T,K)  =\frac{ k(\theta_T - K) \cdot \frac{\partial C}{\partial K} (T,K) + (r + k) \cdot C(T,K) + \frac{\partial C}{\partial T} (T,K)}{\frac{1}{2} K^2 \cdot \frac{\partial^2 C}{\partial K^2} (T,K)}.
\end{equation}
In this setting, the function is given as derivatives of the call-option prices. Equivalent functions can be determined for implied volatility quotes as functions of strike, or log-moneyness. For our implementation, we utilize the implied vol representation in terms of log-moneyness. The formula is slightly longer. Let $d_1,d_2$ be the values from Black's formula, $\mathcal{N}(\cdot)$ the standard normal cumulative distribution function, $n(\cdot)$ the standard normal probability distribution function. The local-volatility function is given by
\begin{equation}
\label{eq:app_lv-vol}
\tilde{\sigma}_{LV}^2(T,x)= \frac{N_1 + N_2}{D_1 + D_2 + D_3},
\end{equation}
with
\begin{align*}
N_1 &= k\theta_i\bigl(\mathcal{N}(d_1) - \mathcal{N}(d_2)\bigr), \\[2mm]
N_2 &= F_{t_0}\, n(d_1)\left[\frac{\sigma(T,x)}{2\sqrt{T}} + \sqrt{T}\cdot\frac{\partial\sigma(T,x)}{\partial T}\right],
\end{align*}
and
\begin{align*}
D_1 &= \frac{F_{t_0}^{T}}{2}\, n(d_1)\left[\frac{1}{\sigma(T,x)\sqrt{T}} + \left(\frac{2d_1}{\sigma(T,x)} - \sqrt{T}\right)\cdot\frac{\partial\sigma(T,x)}{\partial x} + \frac{k(\theta_i - X_0 e^{x})\sqrt{T}}{X_0 e^{x}}\cdot\frac{\partial\sigma(T,x)}{\partial x}\right],\\
D_2 &= \frac{\sqrt{T}\, d_1 d_2}{\sigma(T,x)}\cdot\left(\frac{\partial\sigma(T,x)}{\partial x}\right)^{2} ,\\
D_3 &= \sqrt{T}\cdot\frac{\partial^{2}\sigma(T,x)}{\partial x^{2}},
\end{align*}
where $\tilde{\sigma}_{LV}^2(T,x) = {\sigma}_{LV}^2(T,K)$ with $x = \log(K/S_{t_0})$ is the local vol function in terms of moneyness $x$, and $\sigma(T,x)$ the implied volatility of $S$ for expiry $T$ and moneyness $x$. 
As any local volatility model, computation of \Cref{eq:app_lv} requires a smooth surface of option quotes. Since classical parametrizations like SSVI or SABR are not appropriate in this context due to mean-reversion, we adapt the approach of Drimus and Farkas to construct an interpolator function between option-price slices that preserves arbitrage-free calibration. This interpolator is derived from a necessary relation between arbitrage-free call prices under the LV model.

Let $T_1, T_2, \dots, T_N$ be the expiry grid of the futures curve, and let $\theta_n$, $n \leq N$, denote the corresponding long-term mean as before. Let $C(T,K)$ be the price of a VIX call option under the LV-VIX process. From Proposition 3.1 from ~\cite{drimus2013local}, we have for any $i \leq N$ and $T \in [T_{i-1}, T_i]$,
\begin{align}
\label{eq:diff_prices}
   \left[ 1 - \frac{(T-T_{i-1})\tilde{K}^2}{2}\, \nu_i^2(\tilde{K})\, \e^{-2k(T-T_i)}\, \frac{\partial^2}{\partial K^2} \right] C(T,K) = \e^{-(r+k)(T-T_{i-1})}\, C(T_{i-1},\tilde{K}),
\end{align}
where $\tilde{K} = \theta_i + (K-\theta_i)\e^{k(T-T_i)}$ and $\nu_i(\cdot)$ is a non-vanishing function that encodes the interpolation between the quotes at $T_{i-1}$ and $T_i$. Once the interpolation function is calibrated, this differential equation can be used to obtain option prices $C(T,K)$ for any $T$ and $K$

To calibrate the interpolation function, assume we are given market quotes at $T_{i-1}$ and $T_i$. We first recover a continuum of prices in the strike dimension, i.e.\ the full maps $K \mapsto C(T_{i-1},K)$ and $K \mapsto C(T_i,K)$, by fitting a slice-wise implied volatility parametrization such as SVI or randomized SABR~\cite{zaugg2025volatility}. We then introduce an equally spaced strike grid $K_1 < K_2 < \dots < K_L$, with corresponding values $\tilde{K}_1, \dots, \tilde{K}_L$, and collect the (unknown) values of $\nu_i$ on this grid into a vector $\overrightarrow{\nu}_i = (\nu_{i,1}, \nu_{i,2}, \dots, \nu_{i,L})^\top$.

Evaluating \Cref{eq:diff_prices} at $T = T_i$ and discretizing the derivative on the grid yields the linear system
\begin{align}
\label{eq:linearsystem}
A\, c_{T_i} &= \tilde{c}_{T_{i-1}}, \\
c_{T_i} &= \big(C(T_i,K_1), C(T_i,K_2), \dots, C(T_i,K_L)\big)^\top, \nonumber\\
\tilde{c}_{T_{i-1}} &= \big(C(T_{i-1},\tilde{K}_1), C(T_{i-1},\tilde{K}_2), \dots, C(T_{i-1},\tilde{K}_L)\big)^\top,\nonumber
\end{align}
where
$$
A = \e^{-(r+k)\tau}
\begin{pmatrix}
1 - \dfrac{\tau}{2}\,\tilde K_1^2\, \nu_{i,1}^2 & 0 & \cdots & 0 \\[8pt]
0 & 1 - \dfrac{\tau}{2}\,\tilde K_2^2\, \nu_{i,2}^2 & \cdots & 0 \\[8pt]
\vdots & & \ddots & \vdots \\[4pt]
0 & 0 & \cdots & 1 - \dfrac{\tau}{2}\,\tilde K_L^2\, \nu_{i,L}^2
\end{pmatrix} D_2,
$$
$\tau := T_i - T_{i-1}$, and $D_2$ denotes the finite-difference operator for the second derivative in $K$.

Solving this linear system for $\overrightarrow{\nu}_i$, we recover $\nu_i$ as the piecewise-constant function
$$
\nu_i(K) := \sum_{l=1}^{L-1} \nu_{i,l}\, \mathbf{1}_{K \in [K_l, K_{l+1})}.
$$
This defines the interpolation between the two sets of quotes, and the same procedure can be repeated for every $i \leq N$ to recover the interpolation function across the entire expiry grid. With the linear system above and the interpolation function, we can then recover option prices on a fine grid of times $T$ and strikes $K$, allowing us to compute the derivatives necessary to derive \Cref{eq:app_lv-vol}. This enables us to construct the local volatility function. \Cref{fig:local_vols} shows an example of a local vol function for the VIX.
\begin{figure}[H]
    \centering
    \includegraphics[width=\linewidth]{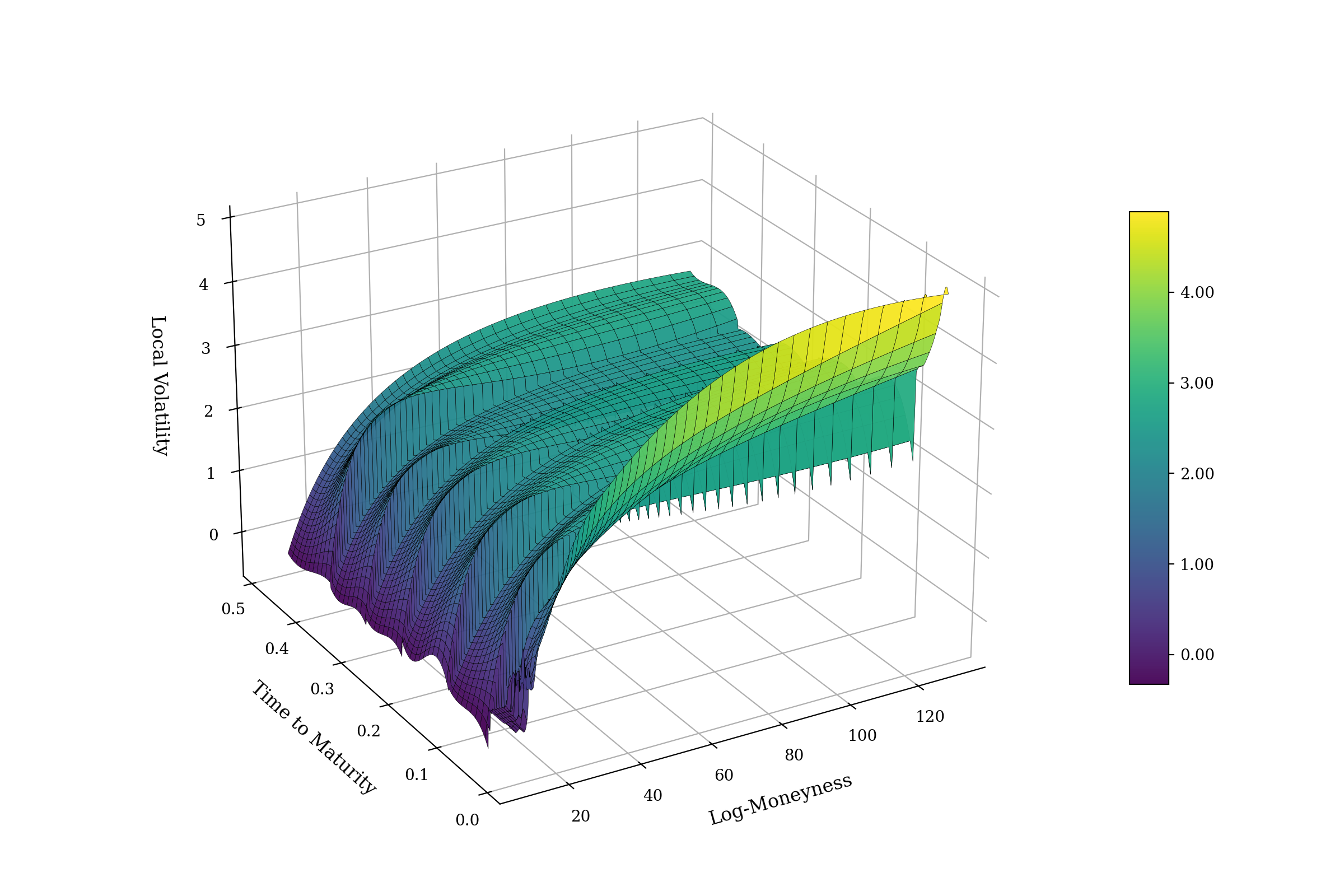}
    \caption{Calibrated VIX Local Vol Surface: $k = 5$. Data from April 27th, 2026}
    \label{fig:local_vols}
\end{figure}
\section{Technical Proofs}
\subsection{Proof of \Cref{lem:epsilon}}
\label{app:proof_epsilon}
Let $q$ be the function at $\psi_{T_F}$ required by the backward equation when $\psi_{T_F+\Theta} = p$. We construct a time-continuous terminal curve such that 
\[\psi_t = \begin{cases}
    p, \quad \text{ if } &t \in [T_F+\delta, T_F+\Theta]\\
    q, \quad \text{ if } &t = T_F\\
    \frac{t - T_F}{\delta} p + \frac{T_F+\delta - t }{\delta}q, \quad \text{ if } &t \in (T_F,T_F+\delta)
\end{cases},\]
where $\delta > 0$ will be determined later. We then compute:
\begin{align*}
    &\norm{\Theta \vix_{T_F}^2 - \mathbb{E}\!\left[\int_{T_F}^{T_F+\Theta} \psi_u(\vix_u)\,\du \;\middle|\; \mathcal{F}_{T_F}\right]}_2 \\
    &\leq \norm{\mathbb{E}\!\left[\int_{T_F}^{T_F+\delta} \psi_u(\vix_u)\,\du \;\middle|\; \mathcal{F}_{T_F}\right]}_2
     + \norm{\Theta \vix_{T_F}^2 - \mathbb{E}\!\left[\int_{T_F+\delta}^{T_F+\Theta} \psi_u(\vix_u)\,\du \;\middle|\; \mathcal{F}_{T_F}\right]}_2\\
     &= \norm{\mathbb{E}\!\left[\int_{T_F}^{T_F+\delta} \psi_u(\vix_u)\,\du \;\middle|\; \mathcal{F}_{T_F}\right]}_2 + \norm{\mathbb{E}\!\left[\int_{T_F}^{T_F+\delta} p(\vix_u)\,\du \;\middle|\; \mathcal{F}_{T_F}\right]}_2
\end{align*}
Since $\psi_u$ and $p$ are polynomials and $\vix_t$ is a polynomial process, there is an $M\geq 0$, such that
\[\norm{\mathbb{E}\!\left[\int_{T_F}^{T_F+\delta} \psi_u(\vix_u)\,\du \;\middle|\; \mathcal{F}_{T_F}\right]}_2 \leq M\delta, \text{ and } \norm{\mathbb{E}\!\left[\int_{T_F}^{T_F+\delta} p(\vix_u)\,\du \;\middle|\; \mathcal{F}_{T_F}\right]}_2 \leq M\delta.\]
Choosing $\delta < \frac{\epsilon}{2M}$ then proves the claim.
\subsection{Proof of \Cref{lem:tech}}
\label{app:proof_lem}
\begin{proof}
Let $(\vix_t)_{t\geq t_0}$ have a generator $\mathcal{L}_\vix = \mu(t,x)\partial_x + \frac{1}{2}\sigma^2(t,x)\partial_{xx}$, where $\mu(t,x) = m_0(t) + m_1(t) x$ and $\sigma^2(t,x) = a_0(t) + a_1(t) x + a_2(t) x^2$ with locally integrable coefficients. On $\mathcal{H}_2 = \text{span}\{1, x, x^2\}$, the generator $\mathcal{L}_t$ has a time-dependent, upper-triangular matrix representation $M(t)$ with diagonal entry $M_{3,3}(t) = 2b_1(t) + a_2(t)$. By the Feynman-Kac backward equation, 
\[\mathbb{E}[X_s^2 \mid X_t = x] = (1, x, x^2) \Phi(t, s) (0, 0, 1)^T,\]  where $\Phi(t, s)$ is the state-transition matrix satisfying $\frac{\partial}{\partial s}\Phi(t, s) = \Phi(t, s)M(s)$ with $\Phi(t, t) = I$. Since $M(\tau)$ is upper-triangular for all $\tau \in [t, s]$, $\Phi(t, s)$ remains upper-triangular, and its $(3,3)$-entry is given explicitly by the scalar linear differential equation solution:
\[
\Phi(t, s)_{3,3} = \exp\left( \int_t^s \left(2b_1(\tau) + a_2(\tau)\right) d\tau \right) > 0.
\]
Thus, the coefficient of $x^2$ is strictly positive, ensuring $\mathbb{E}[X_s^2 \mid X_t = x]$ is degree 2 and non-vanishing for all $s > t$.
\end{proof}
\newpage
\section{Algorithm}
We present a brief algorithmic description to construct the coupling function for general processes using the polynomial approximation method.
\label{app:algo}
\begin{algorithm}
\caption{Construction of the Coupling Function}
\label{alg:coupling-function}
\begin{algorithmic}[1]
\Require Order of projection $m$
\Require Terminal time $T_F$
\Require Drift and diffusion functions $\mu(t,x)$, $\sigma(t,x)$
\Require Time step $\Theta$
\Ensure Sequence of coupling functions $\{\psi_n\}_{n=0}^{N}$

\Function{GetTerminalCurve}{$T_F, m$}
    \State $\mathcal{P} \gets [\,]$ \Comment{Storage for fitted polynomials}
    \For{$k \gets 0$ \textbf{to} $m-1$}
        \State $\bar{M}_{T_F}^k(x) \gets \Call{ComputeConditionalMoment}{T_F, k}$ \Comment{Solve \ref{eq:fc} / integrate in $\tau$}
        \State $p_k \gets \Call{FitPolynomial}{\bar{M}_{T_F}^k(x), m}$ \Comment{Solve \ref{eq:min}}
        \State $\mathcal{P}.\Call{Append}{p_k}$
    \EndFor
    \State $\psi_{T_F} \gets \Call{GetTerminalPsiFromLinearSystem}{\mathcal{P}}$ \Comment{Solve as in \Cref{thm:polynomial}}
    \State \Return $\psi_{T_F}$
\EndFunction

\Statex

\Function{ApplyBackwardEquation}{$\psi, \mu, \sigma, t$}
    \State $\phi \gets \Call{SolveFeynmanKacBackward}{\psi}$ \Comment{Using \Cref{eq:fc_for_psi}}
    \State $\psi_{\mathrm{prev}}(x) \gets \phi(x) - \Theta\left(2x\mu(t,x) + \sigma^2(t,x)\right)$ \Comment{Backward equation (\ref{eq:backwards_main})}
    \State \Return $\psi_{\mathrm{prev}}$
\EndFunction

\Statex

\Function{GetCouplingFunction}{$m, T_F, \mu, \sigma$}
    \State $\psi_{T_F} \gets \Call{GetTerminalCurve}{T_F, m}$ \Comment{Initialize at terminal time}
    \State $\Psi \gets [\psi_{T_F}]$
    \State $t \gets T_F$
    \While{$t > 0$}
        \State $t \gets t - \Theta$
        \State $\psi_t \gets \Call{ApplyBackwardEquation}{\Psi[-1], \mu, \sigma, t}$ \Comment{Step backward in time}
        \State $\Psi.\Call{Append}{\psi_t}$
    \EndWhile
    \State \Return $\Psi$ \Comment{Full sequence of coupling functions}
\EndFunction
\end{algorithmic}
\end{algorithm}
\end{document}